\documentclass[pra,amsmath,amssymb,aps,twocolumn,groupedaddress,nobibnotes
,showpacs]{revtex4-1}

\usepackage{amsthm}
\usepackage{dsfont}

\usepackage{graphicx}
\usepackage{epstopdf}
\usepackage{dcolumn}
\usepackage{bm}
\usepackage{color}
\usepackage{subcaption} 
\usepackage{ragged2e}
\DeclareCaptionJustification{justified}{\justifying}
\newtheorem{theo}{\textsc{Theorem}}

\newtheorem{lemma}[theo]{\textsc{Lemma}}

\usepackage{bbm}
\definecolor{myurlcolor}{rgb}{0,0,0.7}
\definecolor{myrefcolor}{rgb}{0.8,0,0}
\usepackage[pdfusetitle, bookmarks=false,bookmarksnumbered=false,
bookmarksopen=false, breaklinks=false,pdfborder={0 0 0},backref=false,
colorlinks=true, linkcolor=myrefcolor,citecolor=myurlcolor,urlcolor=myurlcolor]{hyperref}

\def\ket#1{| #1 \rangle}
\def\bra#1{\langle #1 |}
\def\dm#1{\left|#1 \right\rangle \left\langle #1 \right|}

\newcommand{\ketbra}[3]{\left| #1 \right\rangle \left\langle #2 \right|_{#3}}

\newcommand{\supp}[1]{\operatorname{supp} #1}

\newcommand{\phip}{\phi^{+}}
\newcommand{\psip}{\psi^{+}}

\newcommand{\phim}{\phi^{-}}
\newcommand{\psim}{\psi^{-}}

\newcommand{\id}{\text{id}}
\newcommand{\K}{\mathcal{K}}
\newcommand{\F}{\mathcal{F}}

\renewcommand{\O}{\mathcal{O}}
\newcommand{\N}{\mathbb{N}}
\newcommand{\C}{\mathbb{C}}

\newcommand{\ZZ}{\mathbb{Z}_2}

\usepackage[dvipsnames]{xcolor}

\newcommand{\bea}{\begin{eqnarray}}
\newcommand{\eea}{\end{eqnarray}}

\makeatletter
\newtheorem*{rep@theorem}{\rep@title}
\newcommand{\newreptheorem}[2]{%
\newenvironment{rep#1}[1]{%
 \def\rep@title{#2 \ref{##1}}%
 \begin{rep@theorem}}%
 {\end{rep@theorem}}}
\makeatother

\begin{document}
\title{Construction of optimal resources for concatenated quantum protocols}
\author{A. Pirker, J. Walln\"ofer, H. J. Briegel and W. D\"ur}
\affiliation{Institut f\"ur Theoretische Physik, Universit\"at Innsbruck, Technikerstr. 21a, A-6020 Innsbruck,  Austria}
\date{\today}

\begin{abstract}
We consider the explicit construction of resource states for measurement-based quantum information processing. We concentrate on special-purpose resource states that are capable to perform a certain operation or task, where we consider unitary Clifford circuits as well as non-trace preserving completely positive maps, more specifically probabilistic operations including Clifford operations and Pauli measurements. We concentrate on $1 \to m$ and $m \to 1$ operations, i.e. operations that map one input qubit to $m$ output qubits or vice versa. Examples of such operations include encoding and decoding in quantum error correction, entanglement purification or entanglement swapping. We provide a general framework to construct optimal resource states for complex tasks that are combinations of these elementary building blocks.  All resource states only contain input and output qubits, and are hence of minimal size. We obtain a stabilizer description of the resulting resource states, which we also translate into a circuit pattern to experimentally generate these states. In particular, we derive  recurrence relations at the level of stabilizers as key analytical tool to generate explicit (graph-) descriptions of families of resource states. This allows us to explicitly construct resource states for encoding, decoding and syndrome readout for concatenated quantum error correction codes, code switchers, multiple rounds of entanglement purification, quantum repeaters and combinations thereof (such as resource states for entanglement purification of encoded states).

\end{abstract}
\pacs{03.67.Hk, 03.67.Lx, 03.67.Pp}
\maketitle


\section{Introduction}

Measurement-based quantum computation \cite{BriegelMeas, BriegelMeas2, Gottesman1999} enables one to perform any quantum computation via a sequence of single-qubit measurements on a large, highly entangled resource state, the 2D cluster state \cite{BriegelCluster}. Hence this specific state is universal for quantum computation. As the 2D cluster state might be very large in terms of the number of qubits, in many situations one is interested in resource states of minimal size for implementing a specific quantum task \cite{RausMeasQCcluster}.


Such a measurement-based approach to quantum information processing has been discussed in various contexts, including quantum computation \cite{BriegelMeas} and quantum communication \cite{ZwergerQC} and was in some cases experimentally demonstrated \cite{LanyonQEC,BarzQEC}. In contrast to a circuit-based approach, no coherent manipulation of quantum information via the application of single- and two-qubit gates is required. One rather needs to prepare certain resource states, which are then manipulated by means of measurements, e.g. by coupling input qubits via Bell measurements to the resource state, similarly as in a teleportation process where however in this case a specific operation, determined by the resource state, is performed. The only sources of errors in such a scheme are imperfect preparation of resource states, and imperfect measurements.

Two main advantages of such a measurement-based approach have been identified \cite{ZwergerQC,ZwergerQEC,ZwergerEPP}: On the one hand, one finds that the acceptable error rates are very high - depending on the task, 10\% noise per particle or more can be tolerated when assuming a noise model where each qubit of the resource state is subjected to single-qubit depolarizing noise. On the other hand, such a measurement-based approach allows for various ways to prepare resource states, including probabilistic (heralded) schemes. This opens the possibility to use probabilistic processes such as parametric-down conversion, error detection, entanglement purification or even cooling to a ground state for resource state preparation.

Depending on the task at hand, different resource states need to be prepared. However, obtaining the explicit resource state or even an efficient preparation procedure for a given task is not straightforward, and explicit constructions are often limited to small system sizes. In turn, any experimental realization requires knowledge of the exact form of resource states and how to prepare them. Also from a theoretical side, knowing the explicit resource state for a specific task allows one to investigate its entanglement features, stability under noise and imperfections and to design optimized ways to generate them with high fidelity, e.g. by means of entanglement purification

Here we provide an explicit construction of resource states for complex tasks that correspond to a composition of elementary building blocks. These building blocks consist of unitary and non-unitary $1 \to m$ or $m \to 1$ operations, including e.g. encoding and decoding in quantum error correction, entanglement swapping and entanglement purification which is a probabilistic process. We develop a general framework for concatenating resource states for different quantum operations according to their stabilizer description and obtain resource states of minimal size which consist only of input and output particles (or only input particles for tasks where there is no quantum output). We provide an efficient and explicit description of families of resource states via recurrence relations in terms of stabilizers for complex quantum operations, and also obtain a representation of these stabilizer states in terms of graph states. This leads directly to an efficient quantum circuit that prepares these states using at most ${\cal O}((n+m)^2)$ commuting two qubit gates for any task with $n$ input and $m$ output systems. Furthermore, this allows one to use several of the methods and techniques developed for graph states \cite{HeinGraph,HeinMulti}, including entanglement purification \cite{DurEPP,DurEPPGraph}, or to analyze their stability under noise and decoherence. Entanglement purification is of particular importance in this context, as this does not only allow one to prepare resource states with high fidelity, but also yields states which can be well described by a local noise model \cite{Wallnoefer}, thereby confirming above mentioned local error model that was used e.g. in \cite{ZwergerQC,ZwergerQEC,ZwergerEPP,ZwergerRepeater}.

The examples we provide include resource states for multiple steps of entanglement purification using a recurrence protocol \cite{Bennett,Deutsch,DurEPP}, encoding, decoding and error syndrome readout for concatenated quantum error correction \cite{Nielsen}, quantum code switchers that allow one to change between different error correction codes (e.g. for storage and data processing), entanglement purification of encoded states or quantum repeater stations for long-distance quantum communication \cite{BriegelRepeater,SangouardRMP}.

This paper is organized as follows. In Sec. \ref{sec:rel} we relate our findings to earlier works in the field and highlight the novelty of our approach. In Sec. \ref{sec:back} we provide some background information on stabilizer states, graph states and Clifford circuits, and give a brief introduction to measurement-based quantum computation and the Choi-Jamiolkowksi isomorphism \cite{Jamiolkowski}. We also settle the notation we use throughout the article in this section. In Sec. \ref{sec:framework} we describe the general framework of concatenated quantum tasks, and present our main technical results to efficiently construct a stabilizer description of corresponding resource states. In Sec. \ref{sec:applications} we present several applications of our method. We provide resource states for multiple steps of entanglement purification using the recurrence protocol of \cite{Deutsch}, concatenated quantum error correction using a generalized Shor code, codes switchers and entanglement purification of encoded states. We summarize and discuss our results in Sec. \ref{sec:discussion}.

\section{Relation to prior work}\label{sec:rel}

Measurement-based quantum computation \cite{BriegelMeas,BriegelMeas2,Gottesman1999,RausMeasQCcluster,NielsenCluster,AliferisComp,ChildsUnified,GrossNovel} is a paradigm where quantum information is processed by measurements only. Certain states, e.g. the 2D cluster state \cite{BriegelCluster}, serve as universal resources \cite{MantriUniversal}. That is, by performing only single qubit measurements, an arbitrary quantum computation can be performed, or equivalently an arbitrary quantum state can be generated. For an introduction and review on measurement-based quantum computation we also refer to \cite{BrowneIntro,JozsaIntro,CampellIntro} and for their  implementation in specific systems to \cite{KwekImpl,KyawImpl,BenjaminProsp}. Despite the probabilistic character of measurements, the desired state is generated deterministically up to local Pauli corrections. There are also special purpose resource states \cite{RausMeasQCcluster,ZwergerQC,ZwergerQEC,ZwergerRepeater} that allow one to perform a specific task or operation. Typically, these special purpose resource states are smaller in size. Such resource states can be constructed in two different ways: First, one may start with a 2D cluster state, and perform all measurements corresponding to Clifford operations. This leaves one with a state of reduced size, which can in principle be determined using the stabilizer formalism \cite{gottesmanphd} or graph-state formalism \cite{HeinGraph}. On the other hand, for circuits that only include Clifford operations and Pauli measurements, one can construct the resource state via the Jamiolokowski isomorphism \cite{Jamiolkowski}, i.e. by applying the circuit to part of a maximally entangled state (see e.g. \cite{RausMeasQCcluster,ZwergerQC,ZwergerQEC,ZwergerRepeater}). This equivalence is also apparent from Theorem 1 in  \cite{RausMeasQCcluster}. In both cases, the stabilizer formalism allows one in principle to efficiently obtain the description of the resource state in terms of its stabilizers. However, taking care of local correction operations and stabilizer update rules is a tedious task, making such a direct approach difficult for complex tasks and operations, especially for operations that act on many qubits and consist of many Clifford operations and measurements.

Measurement-based quantum information processing using special purpose resource states has been investigated in different contexts. In \cite{RausMeasQCcluster} explicit resources for quantum adder and the quantum fourier transform have been constructed. Quantum error correction codes associated with graph states have been proposed in \cite{Schlingemann}. This type of quantum error correction codes especially fit in a measurement-based setting as their resource state corresponds to a graph. Quantum error correction codes where the codewords correspond to graph states were studied e.g. in \cite{SchlingemannGraphCode,Grassl,HeinGraph,HeinMulti}. The concatenation of quantum error correction codes is a standard way to obtain efficient codes for fault-tolerant quantum computation, see e.g. \cite{gottesmanphd}. The measurement-based implementations of quantum error correction codes were studied in \cite{ZwergerQEC,ZwergerQC,BarzQEC,LanyonQEC} where explicit resource states for the repetition code and cluster-ring code were provided. In \cite{DawsonNoiseQuantumComputers,DawsonNoiseClusterState} the practicality of optical cluster-state quantum computation \cite{BrownerOptical,NielsenOptical} using quantum error correction codes was numerically investigated and it was found that scalable optical quantum computing is possible. Furthermore it was also investigated in fault-tolerant quantum computation on cluster-states in \cite{NielsenFault,RaussendorfFault,RaussendorfFault2}. 
In \cite{ZwergerRepeater} explicit resource states for entanglement purification and entanglement swapping, as well as combinations thereof were constructed. In particular, resource states for one and two rounds of entanglement purification using the protocol of \cite{Deutsch} and entanglement purification followed by entanglement swapping. The latter is a building block for a measurement-based quantum repeater, which allows for long-distance quantum communication.

All examples mentioned so far have in common that they construct a resource state for one specific task. Even though these resource states do implement a particular quantum operation, they still lack of a composable description. It is a non-trivial task to concatenate those elementary quantum operations at the level of resource states as already easy examples like two rounds of entanglement purification show. Here we address this problem and provide an explicit construction of resource states for concatenated tasks that are combinations of elementary building blocks. Rather than calculating the required resource state directly, we develop a method to combine and concatenate small building blocks in terms of their stabilizer description via a set of recurrence relations. We would like to emphasize that our construction is applicable not only to unitary Clifford circuits, but also to circuits that contain Pauli measurements. In particular, also probabilistic operations such as entanglement purification can be treated in this way. Pauli measurements can be done beforehand, thereby obtaining resource states of smaller (minimal) size. Furthermore, the result of the Pauli measurement (which determines whether the overall operation is successful and the output states should be kept) can be determined from the results of the incoupling Bell measurements \cite{ZwergerQC,ZwergerEPP,ZwergerRepeater,ZwergerQEC}. This ensures full functionality of circuits, including post-selection based on measurement outcomes, or application of correction operations depending on the encountered error syndrome.

Consider for example a resource state for syndrome read-out and error correction for a concatenated five-qubit code with four concatenation levels. This corresponds to an error correction code of $5^4=625$ qubits, i.e. a resource states of 1250 qubits. The circuit to implement the required error correction operation thus contains several thousand gates, and the direct computation of the resource state from the 2D cluster state or via the Jamiolkowski isomorphism is difficult. With our approach, we make use of the recursive and concatenated structure of the states, and can easily construct the required resource states for decoding and encoding for such concatenated codes in terms of their stabilizers. The combination of decoding (with syndrome readout) and encoding allows then to obtain the resource state for syndrome readout and error correction, or alternatively a code switcher between different codes. In a similar way, one can combine different kinds of elementary building blocks, and obtain resource states for multiple rounds of entanglement purification, entanglement purification of encoded states or quantum repeater stations for encoded quantum information. This allows for full flexibility, and for a broad applicability of our findings.
Following our approach one immediately obtains a stabilizer description of a concatenated quantum operation rather than computing its implementing resource state from scratch. Furthermore, as the construction scheme relies on stabilizers, this description turns out to be especially suited for studying scaling and stability properties of resource states, crucial for experimental implementations. For all relevant examples, we also provide an explicit description of the resulting resource states as graph states. This has the advantage that one obtains directly an efficient way to prepare these states via elementary two-qubit operations. In addition, as entanglement purification protocols for all graph states exist \cite{EPPallGraphs}, one also obtains a way to generate all these states with high-fidelity from multiple copies via entanglement purification.

\section{Background and Notation}\label{sec:back}
In the following we recall some basic notations and results concerning stabilizer states, graph states and measurement-based quantum computation which will be used throughout the paper.

\subsection{Stabilizer states, Clifford group and Graph states}\label{sec:back:stab}

Let $P_n$ denote the $n$ qubit Pauli group, i.e. $P_n$ is the group consisting of all $n-$fold tensor products of the Pauli operators $X,Y,$ and $Z$ as well as the identity. We call an $n$ qubit state $\ket{\psi}$ a stabilizer state if it is stabilized by elements of $P_n$. More precisely, there exist $S_1,..,S_n \in P_n$ such that $S_i \ket{\psi} = \ket{\psi}$ for $i=1,..,n$ \cite{Nielsen}. The stabilizers of $\ket{\psi}$ form a subgroup of $P_n$. \newline
Graph states \cite{HeinGraph,HeinMulti} are specific stabilizer states. Given a mathematical graph $G=(V,E)$, where $V$ denotes the set of vertices and $E$ the set of edges, the associated graph state $\ket{G}$ is stabilized by the operators
\begin{align}
K_{a} = X^{(a)} \prod_{\lbrace a,b \rbrace \in E} Z^{(b)}
\end{align}
where the superscripts in brackets indicate on which Hilbert space the Pauli operator acts. Hence the graph state $\ket{G}$ is the common $+1$ eigenstate of the family of operators $\lbrace K_{a} \rbrace_{a \in V}$. Alternatively, the graph state $\ket{G}$ can be generated via
\begin{align}
\ket{G} = \prod_{\lbrace a,b \rbrace \in E} U^{\lbrace a,b \rbrace} \ket{+}^{\otimes V}
\end{align}
where $U^{\lbrace a,b \rbrace} = \dm{0} \otimes \id + \dm{1} \otimes Z$ is a controlled $Z-$gate. Notice that this implies an efficient preparation procedure for all graph states, i.e. knowing the graph state description of a state provides one with a way to prepare the state with at most quadratically many commuting two-qubit gates.

We call two graph states $\ket{G}$ and $\ket{G'}$ \emph{local unitary} equivalent (LU equivalent) if there exist unitaries $U_i$ such that $U_1\otimes \dots \otimes U_n \ket{G} = \ket{G'}$. \newline
An important group of unitaries is the so-called Clifford group. The Clifford group is the set of all $n$ qubit unitaries $U$ such that $U P_n U^\dagger = P_n$, or, in other words, the Clifford group is the normalizer of the Pauli group. An important result which we will use here frequently is that any stabilizer state is \emph{local Clifford} equivalent (LC equivalent) to a graph state \cite{Nest} and that those local Clifford operations can be determined efficiently. This equivalence can be most easily derived in terms of the binary representation of stabilizers of a stabilizer state. \newline
Finally, we denote the four Bell-basis states by
\begin{align}
\ket{\phip} = (\ket{00} + \ket{11})/\sqrt{2}, \\
\ket{\phim} = (\ket{00} - \ket{11})/\sqrt{2}, \\
\ket{\psip} = (\ket{01} + \ket{10})/\sqrt{2}, \\
\ket{\psim} = (\ket{01} - \ket{10})/\sqrt{2}.
\end{align}

\subsection{Measurement-based quantum computation and the Jamiolkowski isomorphism}\label{sec:back:meas}

In measurement-based quantum computation a specific quantum operation is realized via a sequence of single-qubit measurements on an entangled state. It has been shown in \cite{BriegelMeas2} that the 2D cluster state \cite{BriegelCluster} is universal for quantum computation \cite{MantriUniversal}. The 2D cluster state is a graph state associated with a two-dimensional square lattice.

This enables a correspondence between a quantum operation and a specific sequence of measurements on the 2D cluster, as those measurements implement the quantum operation in a measurement-based way. However, the outcomes of measurements are random, leading to Pauli correction operations. For general quantum circuits, this implies that measurements need to be modified depending on previous measurement results, and quantum information processing takes place in a sequential way. Nevertheless, by using a sufficiently large resource state, one can deterministically implement an arbitrary quantum circuit following this approach. Determinism in measurement-based quantum computation was addressed in more depth in \cite{DanosDeter,BrowneDeter}.

The 2D cluster state is a universal resource and can thus be used to simulate any quantum circuit. However, there can be a large overhead in terms of the number of auxiliary systems. So if one is interested in a specific quantum task, it might be beneficial to consider a special-purpose resource state that can be used to realize a specific operation, and find a state of minimal size or complexity \cite{RausMeasQCcluster}. The Jamiolkowski isomorphism \cite{Jamiolkowski} establishes a one-to-one correspondence between completely positive maps and quantum states. More precisely, for every quantum operation $\O$ there exists a mixed state $E_{\O}$ (which we will refer to as \emph{resource state}) such that $E_{\O}$ allows one to {\em probabilistically} implement the operation $\O$ via Bell-measurements at the input qubits of the resource state, see Fig. \ref{fig:back:measimpl}.

\begin{figure}[h!]
\scalebox{0.7}{
\includegraphics{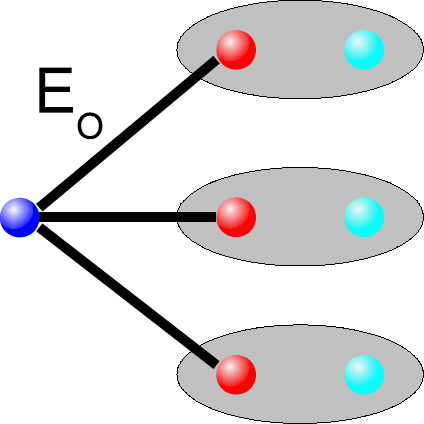}
}
\caption[h!]{Measurement-based implementation of the quantum operation $O$ mapping $n=3$ qubits to $m=1$ qubit via its Jamiolkowksi state $E_O$. The input qubits of the Jamiolkowski state are red, the output qubit is blue and the ellipses indicate Bell-measurements.} \label{fig:back:measimpl}
\end{figure}

For unitary operations $U$, or unitary operations followed by projective measurements on some of the output qubits, the corresponding resource state is pure, $\ket{E_{U}}$. We only consider such situations in the following.

In this case, the Jamiolkowski state $\ket{E_{\O}}$ for the quantum operation $\O$ is given by
\begin{align}
\ket{E_{\O}} = (I_A \otimes \O_B) \bigotimes\limits^{n}_{i=1} \ket{\phip}_{A_i B_i},
\end{align}
(see also Theorem 1 in \cite{RausMeasQCcluster}), where $n$ denotes the number of input qubits of $\O$, and $A$ and $B$ denote input and output qubits respectively. For an $n \to m$ operation, i.e. an operation acting on $n$ input qubits and producing $m$ output qubits, the resource state is of size $n+m$.
The processing of quantum information now takes place by coupling the qubits of the input states to the input qubits of the resource state via Bell measurements, similarly as in teleportation. Depending on the outcome of the Bell measurements, the output state is then given by a state, where first some Pauli operations $\otimes_n \sigma_{i_k}$ (determined by the outcomes of the Bell measurements) act, followed by the application of the desired operation $\O$. In general, the Pauli operations and $\O$ do not commute, resulting in a probabilistic implementation of $\O$. In particular, if all outcomes of the Bell measurements are given by $\ket{\phi^+}$ (which happens with probability $1/4^n$), the desired operation is applied.

For specific operations $\O$, including all Clifford circuits (unitary operations from the Clifford group and Pauli measurements), the Pauli operations can be corrected, and one obtains a {\em deterministic} implementation of the map in this case. All circuits we consider throughout the paper are of this form.




It is also straightforward to concatenate different quantum tasks, as the quantum computation is done via Bell-measurements at the input qubits. In particular, one can combine resource states via Bell-measurements on the respective inputs and outputs, see Fig. \ref{fig:construction:all}.


\begin{figure}[h!]
\begin{subfigure}{\linewidth}
\scalebox{0.72}{
\includegraphics{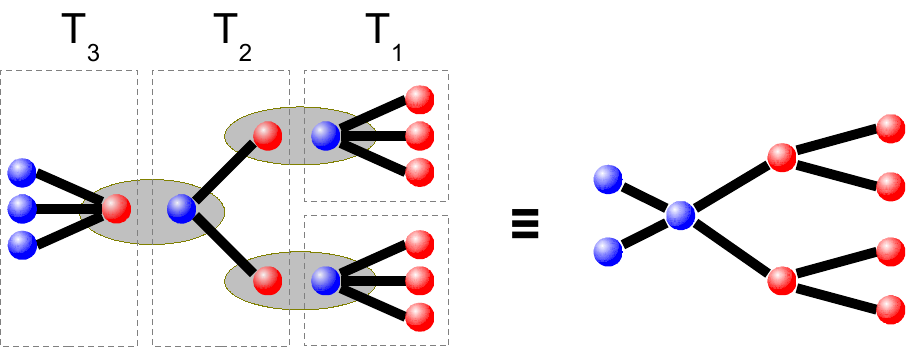}
}
\end{subfigure}\\[1ex]
\begin{subfigure}{\linewidth}
\scalebox{0.72}{
\includegraphics{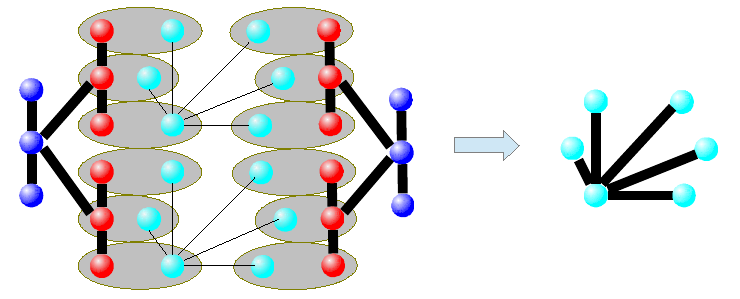}
}
\end{subfigure}
\caption{Construction and usage of the resource state for the concatenated quantum task \emph{entanglement purification at a logical level}. The first figure illustrates the concatenation of several different quantum tasks. In this example, the quantum tasks \emph{decoding in 3-qubit bit-flip code} ($T_1$), \emph{one round of the entanglement purification protocol \cite{Deutsch}} ($T_2$) and \emph{encoding in 3-qubit bit-flip code} ($T_3$) are concatenated. On the right hand-side an LC-equivalent graph state for performing all three quantum tasks in one go is shown. This resource state implements one round of entanglement purification at a logical level, i.e. without decoding. The second figure depicts the usage of the resource state in a measurement-based implementation of the concatenated quantum task \emph{entanglement purification at a logical level}.}
\label{fig:construction:all}
\end{figure}

Because those Bell-measurements might, in general, yield other outcomes than $\ket{\phip}$ one would need to deal in sequential implementations of concatenated quantum tasks with Pauli byproduct operators. We emphasize that within the framework proposed here these coupling measurements are done virtually only, thus enabling deterministic generation of resource states for complex quantum tasks. That is, the intermediate output-input qubits only appear virtually, and are not part of the resulting resource state that contains input and output qubits of the general task. This is similar as for qubits which are measured in the Pauli basis in a Clifford circuit: also in this case, the measurement can be done beforehand on qubits in the resource state, thereby leading to a state of reduced size that only contains input and output qubits. The result of Pauli measurements on beforehand measured (virtual) systems is determined by the outcomes of the incoupling Bell measurements see e.g. \cite{ZwergerRepeater} for entanglement purification. This leads to a possible reinterpretation of the outcomes at the read-in.

In the following we will use at some points of the paper the notation:
\begin{align}
\sigma_{0,0} = \id, \quad \sigma_{0,1} = X, \quad \sigma_{1,0} = Z, \quad \sigma_{1,1} = Y.
\end{align}
We denote by $n \stackrel{\O}{\to} m$ that the quantum operation $\O$ maps $n$ input qubits to $m$ output qubits. Furthermore, all resource states we are concerned with here are stabilizer states.

\section{Framework} \label{sec:framework}

In this section we present a general framework to construct the stabilizers of resource states for concatenated quantum tasks.

\subsection{Basic observation}

Suppose we implement the quantum operation $1 \stackrel{\O}{\to} m$ via its Jamiolkowski state $\ket{\psi^{\O}}$ in a measurement-based way, for example encoding quantum information in a quantum error correcting code. We rewrite $\ket{\psi^{\O}}$ as
\begin{align}
\ket{\psi^{\O}} = \ket{+}_{in} \ket{G_{0}}_{out} + \ket{-}_{in} \ket{G_{1}}_{out} = \sum_i \ket{i^{x}}_{in} \ket{G_{i}}_{out} \label{eq:observation:plusminus}
\end{align}
where $\ket{+}_{in} = \ket{0^{x}}_{in}$ and $\ket{-}_{in} = \ket{1^{x}}_{in}$ denote the eigenstates of $X$, $in$ the input qubit and $out$ the system of output qubits.
Observe that the states $\ket{G_{0}}_{out}$ and $\ket{G_{1}}_{out}$ are not normalized. Furthermore assume there exist two classes of operators, which we call \emph{auxiliary} operators $K$ and $F$, which satisfy the equations:
\begin{align}
K \ket{G_{i}}_{out} &= \ket{G_{i \oplus 1}}_{out}, \label{def:k:aux}\\
F \ket{G_{i}}_{out} &= (-1)^{i} \ket{G_{i}}_{out}. \label{def:f:aux}
\end{align}
The operator $K$ serves as logical $X$ and the operator $F$ as a logical $Z$ operator for the states $\ket{G_{i}}_{out}$. We denote by $\K$ and $\F$ the sets of all operators satisfying (\ref{def:k:aux}) and (\ref{def:f:aux}) respectively. Using (\ref{def:k:aux}) and (\ref{def:f:aux}) we immediately observe that the operators
\begin{align}
Z &\otimes K \label{framework:kstab}\\
X &\otimes F \label{framework:fstab}
\end{align}
stabilize $\ket{\psi^{\O}}$, i.e. $(Z \otimes K)\ket{\psi^{\O}} = \ket{\psi^{\O}}$ and $(X \otimes F)\ket{\psi^{\O}} = \ket{\psi^{\O}}$ for $K \in \K$ and $F \in \F$. So we infer that some stabilizers of the resource state $\ket{\psi^{\O}}$ are given by $\lbrace Z \otimes K \rbrace_{K \in \K} \cup \lbrace X \otimes F \rbrace_{F \in \F}$. \

According to the following argument the stabilizers can always be brought to this form: If there are at least two generators of the stabilizer group with different Pauli operators on the first qubit we can construct a set of generators that are of form (\ref{framework:kstab}) and (\ref{framework:fstab}) by picking different elements of the stabilizer group. This can be done by multiplying two of the known stabilizers together to get a new one.
Furthermore, all meaningful (i.e. entangled, so they actually can implement an operation) resource states have at least two stabilizers with different Pauli operators as every set of stabilizers is local Clifford equivalent to a set of graph state stabilizers \cite{Nest}.
So, in the graph state picture there is an edge from the first qubit to at least one of the other qubits. That means the stabilizers can always be brought to the required form if the input qubit is entangled with the other qubits.

In the following we will construct the stabilizers of a resource state implementing a concatenated quantum task, such as concatenated quantum error correction or entanglement purification at a logical level. Suppose we want to concatenate the quantum operations $1 \stackrel{\O}{\to} m$ and $1 \stackrel{\O'}{\to} n$, where the latter acts on each output particles of the former, i.e. the resulting quantum operation is $\mathcal{O'}^{\otimes m} \circ \mathcal{O}$, and both operations are implemented in a measurement-based way \footnote{We emphasize that the same procedure applies to quantum operations of the form $m \stackrel{\O}{\to} 1$ and $n \stackrel{\O'}{\to} 1$, i.e. taking multiple qubits as input and producing a single qubit output. Examples thereof are concatenated entanglement purification protocols or decoding a logical qubit.}. Examples thereof are encoding a single qubit into a quantum error correction code. Hence we have to combine the resource states $\ket{\psi^{\O}}$ and $\ket{\psi^{\O'}}^{\otimes m}$. Expressing $\ket{\psi^{\O}}$ and $\ket{\psi^{\O'}}$ as in (\ref{eq:observation:plusminus}) yields
\begin{align}
\ket{\psi^{\O}} &= \sum\limits^1_{i=0} \ket{i^{x}}_{in} \ket{G^{1}_{i}}_{out}, \\
\ket{\psi^{\O'}} &= \sum\limits^1_{i=0} \ket{i^{x}}_{in'} \ket{G^{n}_{i}}_{out'}.
\end{align}
Combining $\ket{\psi^{\O}}$ and $\ket{\psi^{\O'}}^{\otimes m}$ is done via Bell-measurements, see Sec. \ref{sec:back:meas}, between $out$ and $in'$, see Fig. \ref{fig:basic:observarion}.
\begin{figure}[h!]
\scalebox{1.1}{
\includegraphics{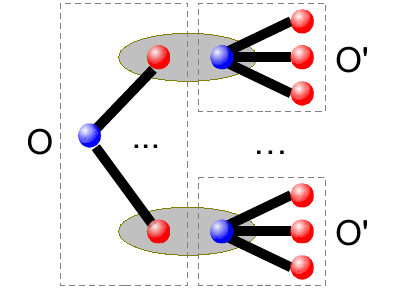}
}
\caption[h!]{Graphical illustration of the construction of the resource state for the concatenated quantum tasks $O'^{\otimes m} \circ O$. The ellipsis indicated Bell-measurements.} \label{fig:basic:observarion}
\end{figure}
Because we are interested in the construction of the resource state for implementing the composite $\mathcal{O'}^{\otimes m} \circ \mathcal{O}$ rather than their sequential execution we assume that the (virtual) coupling Bell-measurements reveal simultaneously the $\ket{\phip}$ outcome. It is straightforward to define \emph{connecting} functions
\begin{align}
\alpha_{i}(k_1,\ldots,k_m) = {}^{\otimes m}\! \langle \phip \ket{k^x_1,\ldots,k^x_m}_{in'^m} \ket{G^{1}_{i}}_{out} \label{def:confunc}
\end{align}
for $i = 0,1$ where $\ket{k^x_1,\ldots,k^x_m}_{in'^m} = \bigotimes^m_{l=1} \ket{k^x_l}_{in'_l}$. To summarize, the resource state for implementing the quantum operation $\mathcal{O'}^{\otimes m} \circ \mathcal{O}$ by denoting $\ket{\Phi} = \bigotimes^m_{l=1} \ket{\phip}_{in'_l, out_l}$ is given by
\begin{align}
\ket{\psi^{\mathcal{O'}^{\otimes m} \circ \mathcal{O}}} & = \frac{1}{2^m} \bra{\Phi} \psi^{\O} \rangle \ket{\psi^{\O'}}^{\otimes m} \notag \\
&= \frac{1}{2^m} \bra{\Phi} \sum_{i} \ket{i^{x}}_{in} \sum^1_{k^x_1,..,k^x_m=0} \ket{k^x_1,\ldots,k^x_m}_{in'^m}  \notag \\
& \hspace{1cm} \otimes \ket{G^{1}_{i}}_{out} \ket{G^{n}_{k_1},\ldots,G^{n}_{k_m}}_{out'^m} \notag \\
&= \frac{1}{2^m} \sum_{i} \ket{i^{x}}_{in} \sum^1_{k_1,..,k_m=0} \alpha_{i}(k_1,\ldots,k_m) \notag \\
& \;\;\;\;\;\;\;\;\;\;\;\;\;\;\;\;\;\;\;\;\;\;\;\;\;\; \ket{G^{n}_{k_1},\ldots,G^{n}_{k_m}}_{out'^m} \label{framework.eq.1}
\end{align}
where $\ket{G^{n}_{k_1},\ldots,G^{n}_{k_m}}_{out'^m} = \bigotimes^m_{l=1} \ket{G^{n}_{k_l}}_{out'_l}$. We observe that (\ref{framework.eq.1}) is again of the form (\ref{eq:observation:plusminus}). In the rest of this section we elaborate on the construction of the auxiliary operators of $\ket{\psi^{\mathcal{O'}^{\otimes m} \circ \mathcal{O}}}$ in terms of the auxiliary operators of $\ket{\psi^{\O'}}$. We emphasize that the same techniques apply to the concatenation of quantum operations with a single qubit output.

\subsection{Main results}

In the following we denote concrete auxiliary operators of $\ket{\psi^{\mathcal{O'}^{\otimes m} \circ \mathcal{O}}}$ by $K_{n+1}$ and $F_{n+1}$, i.e. $K_{n+1} \in \K_{n+1}$ and $F_{n+1} \in \F_{n+1}$. The theorem below relates the auxiliary operator $K_{n+1}$ of $\ket{\psi^{\mathcal{O'}^{\otimes m} \circ \mathcal{O}}}$ to auxiliary operators $K_{n} \in \K_{n}$ and $F_{n} \in \F_{n}$ of $\ket{\psi^{\O'}}$.
\begin{theo}[Recurrence relation for $K_{n+1}$]\label{thm.koperators}
Let $K_1 = a \bigotimes^m_{k=1} \sigma_{i_k,j_k}$ with $a \in \lbrace \pm 1, \pm i \rbrace$, $i_k \in \lbrace 0,1 \rbrace$ and $j_k \in \lbrace 0,1 \rbrace$ be the $K-$type operator of $\ket{\psi^{\mathcal{O}}}$. \newline
Then there exist $\gamma \in \ZZ^m$, $\delta \in \ZZ^m$ and $c \in \C$ where $|c|=1$ such that the auxiliary operator $K_{n+1}$ of $\ket{\psi^{\mathcal{O'}^{\otimes m} \circ \mathcal{O}}}$ is given by $K_{n+1} = c \bigotimes\limits^m_{r=1} (F_r)^{\delta_r} \bigotimes\limits^m_{s=1} (K_s)^{\gamma_s}$ where $\gamma_s = i_s$, $\delta_s = j_s$, $K_{s} \in \K_{n}$ and $F_{s} \in \F_{n}$ for all $s$ and $c = a \prod\limits^m_{l=1} i^{i_l j_l}$.
\end{theo}
We sketch the proof of Theorem \ref{thm.koperators} as follows: First one observes that the new auxiliary operator $K_{n+1}$ must satisfy (\ref{def:k:aux}), i.e. $K_{n+1} \ket{G^{n+1}_{i}} = \ket{G^{n+1}_{i \oplus 1}}$. The application of $K_{n+1} = c \bigotimes\limits^m_{r=1} (F_r)^{\delta_r} \bigotimes\limits^m_{s=1} (K_s)^{\gamma_s}$ to $\ket{G^{n+1}_{i}} = \sum_{k_1,..,k_m} \alpha_{i}(k_1,\ldots,k_m) \ket{G^{n}_{k_1} ... G^{n}_{k_m}}$ leads to $c (-1)^{\sum_r \delta_r k_r} \alpha_{i}(k_1 \oplus \gamma_1,..,k_m \oplus \gamma_m) =\alpha_{i \oplus 1}(k_1,..,k_m)$. From the definition of the connecting functions $\alpha_i$ (\ref{def:confunc}), the expansion of $K_1$ in the Pauli basis and the well-known relationship $(\sigma_{i,j} \otimes \id) \ket{\phip} = (-1)^{ij}(\id \otimes \sigma_{i,j}) \ket{\phip}$ follows the claim. We provide a detailed proof in Appendix \ref{app:proofs}. \newline
Theorem \ref{thm.koperators} establishes a recurrence relation for auxiliary operators of type $K$ for concatenated quantum tasks. We provide a similar theorem for the auxiliary operators of type $F$.
\begin{theo}[Recurrence relation for $F_{n+1}$]\label{thm.foperators1}
Let $F_1 = a \bigotimes^m_{k=1} \sigma_{i_k,j_k}$ with $a \in \lbrace \pm 1, \pm i \rbrace$, $i_k \in \lbrace 0,1 \rbrace$ and $j_k \in \lbrace 0,1 \rbrace$ be the $F-$type operator of $\ket{\psi^{\mathcal{O}}}$. \newline
Then there exist $\epsilon \in \ZZ^m$, $\eta \in \ZZ^m$ and $c \in \C$ where $|c|=1$ such that the auxiliary operator $F_{n+1}$ of $\ket{\psi^{\mathcal{O'}^{\otimes m} \circ \mathcal{O}}}$ is given by $F_{n+1} = c \bigotimes\limits^m_{r=1} (F_r)^{\eta_r} \bigotimes\limits^m_{s=1} (K_s)^{\epsilon_s}$ where $\epsilon_s = i_s$, $\eta_s = j_s$, $K_{s} \in \K_{n}$ and $F_{s} \in \F_{n}$ for all $s$ and $c = a \prod\limits^m_{l=1} i^{i_l j_l}$.
\end{theo}
The proof is similar to the proof of Theorem \ref{thm.koperators}: $F_{n+1}$ must satisfy (\ref{def:f:aux}), i.e. $F_{n+1} \ket{G^{n+1}_{i}} = (-1)^{i} \ket{G^{n+1}_{i}}$. Hence applying $F_{n+1}$ to $\ket{G^{n+1}_{i}}$ implies the condition $c (-1)^{\sum_r \eta_r k_r} \alpha_{i}(k_1 \oplus \epsilon_1,..,k_m \oplus \epsilon_m) = (-1)^{i} \alpha_{i}(k_1,..,k_m)$. Again, from the definition of $\alpha_i$ (\ref{def:confunc}), and the expansion of $F_1$ the claim follows. The detailed proof is provided in Appendix \ref{app:proofs}. \newline
Thus, Theorem \ref{thm.koperators} and \ref{thm.foperators1} show that the sets $\K_{n+1}$ and $\F_{n+1}$ are given by
\begin{align}
\K_{n+1} = \bigcup\limits_{1 \leq l \leq |\K_1|} \lbrace c \bigotimes\limits^m_{r=1} (F_r)^{\delta^{(l)}_r} & \bigotimes\limits^m_{s=1} (K_s)^{\gamma^{(l)}_s} : \notag \\
& F_r \in \F_{n}, K_s \in \K_{n} \rbrace, \label{eq:k:fam} \\
\F_{n+1} = \bigcup\limits_{1 \leq l \leq |\F_1|} \lbrace c \bigotimes\limits^m_{r=1} (F_r)^{\eta^{(l)}_r} & \bigotimes\limits^m_{s=1} (K_s)^{\epsilon^{(l)}_s} : \notag \\
& F_r \in \F_{n}, K_s \in \K_{n} \rbrace, \label{eq:f:fam}
\end{align}
where the families $\lbrace \gamma^{(l)} \rbrace_{1 \leq l \leq |\K_1|}$, $\lbrace \delta^{(l)} \rbrace_{1 \leq l \leq |\K_1|}$ and $\lbrace \epsilon^{(l)} \rbrace_{1 \leq l \leq |\F_1|}$, $\lbrace \eta^{(l)} \rbrace_{1 \leq l \leq |\F_1|}$ denote the decompositions in the Pauli basis of the initial auxiliary operators in $\K_1$ and $\F_1$ respectively. \newline
The sets $\K_{n+1}$ and $\F_{n+1}$ enable us to provide a complete set of recurrence relations for auxiliary operators via (\ref{eq:k:fam}) and (\ref{eq:f:fam}). Recall that the auxiliary operators immediately translate to stabilizers via (\ref{framework:kstab}) and (\ref{framework:fstab}). In general the sets $\K_{n+1}$ and $\F_{n+1}$ will contain too many auxiliary operators depending on $|\K_1|$ and $|\F_1|$ and $|\K_n|$ and $|\F_n|$. This is rather obvious as all initial auxiliary operators in $\K_1$ and $\F_1$ enable the application of Theorem \ref{thm.koperators} and \ref{thm.foperators1}. The new stabilizers uniquely describe the resulting state because we can choose a sufficient number (i.e. the number of qubits in the resulting state) of linear independent stabilizers of the form (\ref{framework:kstab}) and (\ref{framework:fstab}) from the sets $\K_{n+1}$ and $\F_{n+1}$. For the examples we are concerned with in the following sections this follows immediately from the construction of the new stabilizers using Theorem \ref{thm.koperators} and \ref{thm.foperators1}. We address this issue in detail for those cases in Appendix \ref{app:lin}. Furthermore, we show that this method always provides a full set of independent stabilizers in Appendix \ref{app:numberofstabilizers}.\newline
We conclude the technical results by providing a theorem concerning the concatenation of single qubit output with single qubit input quantum tasks, crucial for the measurement-based implementation of code switchers and quantum repeaters \cite{BriegelRepeater}. The setting is illustrated in Fig. \ref{fig:thm:coupling}.

\begin{figure}[h!]
\scalebox{1}{
\includegraphics{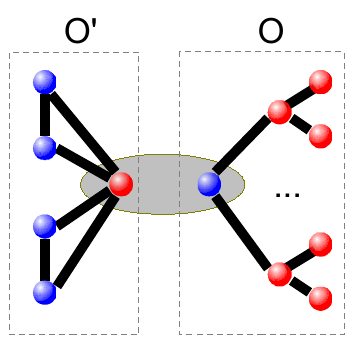}
}
\caption[h!]{Illustration of the concatenation of single qubit output with single qubit input quantum tasks via a Bell-measurement. Here the single qubit output task $O$ is concatenated with the single qubit input task $O'$, e.g. $O$ might be a decoding procedure in a quantum error correction and $O'$ might be an encoding procedure into another quantum error correction code.} \label{fig:thm:coupling}
\end{figure}

\begin{theo}[Coupling of resource states]\label{thm.coupling}
Suppose we are given two resource states, $\ket{\psi_1} = \ket{+}_{in} \ket{G_0} + \ket{-}_{in} \ket{G_1}$ and $\ket{\psi_2} = \ket{+}_{out} \ket{H_0} + \ket{-}_{out}\ket{H_1}$, implementing a single-qubit input and single-qubit output quantum task respectively. Furthermore assume that we are provided the sets of auxiliary operators $\K_1, \F_1$ and $\K_2, \F_2$, i.e.
\begin{align}
K_1 \ket{G_i} = \ket{G_{i \oplus 1}} \quad \text{ and } \quad F_1 \ket{G_{i}} = (-1)^{i} \ket{G_{i}}, \\
K_2 \ket{H_i} = \ket{H_{i \oplus 1}} \quad \text{ and } \quad F_2 \ket{H_{i}} = (-1)^{i} \ket{H_{i}}.
\end{align}
where $K_i \in \K_i$ and $F_i \in \F_i$. Then the stabilizers of the concatenated quantum task, i.e. the resource state obtained by connecting $\ket{\psi_2}$ and $\ket{\psi_1}$ through a Bell-measurement between $in$ and $out$, are given by
\begin{align}
\lbrace K_1 \otimes K_2: \, & K_1 \in \K_1 , K_2 \in \K_2 \rbrace \notag \\
& \cup \lbrace F_1 \otimes F_2: \, F_1 \in \F_1, F_2 \in \F_2 \rbrace
\end{align}
\end{theo}
\begin{proof}
First we observe that $\ket{\phip} = (\ket{++} + \ket{--})/ \sqrt{2}$. Hence we find for the state after connecting $\ket{\psi_1}$ and $\ket{\psi_2}$ via the Bell-measurement
\begin{align}
\frac{1}{2} & \ketbra{\phip}{\phip}{in,out} \ket{\psi_1} \ket{\psi_2} \notag \\
& = \frac{1}{2} (\ket{G_0}\ket{H_0} + \ket{G_1}\ket{H_1}) \ket{\phip} =: \ket{\psi}.
\end{align}
It follows for $K_i \in \K_i$ and $F_i \in \F_i$ that
\begin{align}
(K_1 \otimes K_2) \ket{\psi} &= (K_1 \otimes K_2) (\ket{G_0}\ket{H_0} + \ket{G_1}\ket{H_1})/2 \notag \\
& = (\ket{G_1}\ket{H_1} + \ket{G_0}\ket{H_0})/2 = \ket{\psi}, \\
(F_1 \otimes F_2) \ket{\psi} &= (F_1 \otimes F_2) (\ket{G_0}\ket{H_0} + \ket{G_1}\ket{H_1})/2 \notag \\
& = (\ket{G_0}\ket{H_0} + (-1)^2\ket{G_1}\ket{H_1})/2 = \ket{\psi},
\end{align}
which shows that $K_1 \otimes K_2$ and $F_1 \otimes F_2$ stabilize $\ket{\psi}$ as claimed.
\end{proof}

\subsection{Dealing with byproduct operators}\label{sec:frm:by}

Finally we have to deal with Pauli byproduct operators due to the Bell-measurements at the read-in of the resource state. Because all quantum operations we are concerned with here belong to the Clifford group, the Pauli byproduct operators at the input qubits propagate through the resource state leading to (possibly different) Pauli operations at the output qubits, see Sec. \ref{sec:back:stab}. Thus it is left to determine how those Pauli errors propagate through a concatenated resource state which we discuss below.

If the resource states $\ket{\psi^\O}$ and $\ket{\psi^{\O'}}$ implementing the quantum operations $\O$ and $\O'$ are concatenated to obtain the resource state $\ket{\psi^{\O'^{\otimes m} \circ O}}$ implementing $\O'^{\otimes m} \circ \O$, the Pauli byproduct operators due to the read-in Bell-measurements need to be propagated through $\ket{\psi^{\O'^{\otimes m} \circ O}}$. This is done as follows: the Pauli byproduct operator $\sigma_{\alpha,\beta}$ at the read-in of $\ket{\psi^\O}$ translates to an effective Pauli error $\sigma_{f(\alpha,\beta)}$ on the output of $\ket{\psi^\O}$. Because $\ket{\psi^\O}$ and $\ket{\psi^{\O'}}^{\otimes m}$ were connected via a virtual Bell-measurement with $\ket{\phip}$ outcome, the Pauli error $\sigma_{f(\alpha,\beta)}$ at the output of $\ket{\psi^{\O}}$ is also the effective Pauli error on the input of $\ket{\psi^{\O'}}^{\otimes m}$. Hence one obtains the resulting Pauli error up to a global phase on the output of $\ket{\psi^{\O'^{\otimes m} \circ \O}}$ by propagating $\sigma_{f(\alpha,\beta)}$ through $\ket{\psi^{\O'}}^{\otimes m}$ leading to a deterministic correctable Pauli error $\sigma_{g(f(\alpha,\beta))}$ at the output of $\ket{\psi^{\O'^{\otimes m} \circ \O}}$.

\section{Applications}\label{sec:applications}

In Sec. \ref{sec:framework} we showed theorems which allow one to construct auxiliary operators for concatenated quantum tasks, and determine the stabizers of the resulting resource states. Now we provide applications of those theorems to different tasks in quantum communication and computation. In particular, we use Theorems \ref{thm.koperators}-\ref{thm.coupling} to construct stabilizers and graph state representation of resource states for measurement-based implementations of multiple rounds of entanglement purification protocols, for quantum error correction including encoding, decoding and syndrome readout for concatenated quantum codes, as well as for code switchers. Finally we construct resource states for the measurement-based implementation of entanglement purification protocols at a logical level.

\subsection{Bipartite entanglement purification protocols}

Bipartite entanglement purification protocols are used to distill high-fidelity entangled states, ideally a perfect Bell-pair, from a set of noisy copies by means of local operations and classical communication. 
Entanglement purification is an important primitive in quantum information processing \cite{DurEPP}, and constitutes a possible way to prepare states with high fidelity, both locally where they have the role of resource states to perform certain operations or tasks \cite{Wallnoefer}, as well as in a distributed, non-local way, where entanglement purification is used as a central element in long-distance quantum communication protocols, the quantum repeater \cite{BriegelRepeater}. Entanglement purification protocols exist for all graph states \cite{EPPallGraphs}, and the methods we develop  here are also applicable to such protocols. That is, one can construct the corresponding resource states to perform multiple rounds of multi-party entanglement purification in a similar way as outlined below for bipartite states.

Several different protocols for bipartite entanglement purification have been proposed \cite{Deutsch,Bennett,DurEPP,AschauerThesis}. Here we provide a detailed analysis of a measurement-based implementation of the entanglement purification protocol of \cite{Deutsch}, which we refer to as DEJMPS protocol in the following. The DEJMPS protocol is a recurrence-type entanglement purification protocol that operates on two noisy pairs and produces probabilistically one pair with improved fidelity. This is achieved by applying at each site, referred to as Alice and Bob, a certain two-qubit operation between the two pairs, followed by a local measurement of the second pair. Depending on the measurement outcome, the remaining pair is either kept or discarded, a step for which two-way classical communication is necessary. This constitutes the basic purification step, where only Clifford operations and Pauli measurements are involved and hence a three-qubit resource state with two input and one output particle can be found for a measurement-based implementation \cite{ZwergerEPP,ZwergerRepeater}. This basic purification step (also referred to as purification round) may be applied in an iterative manner, using the output pairs of the previous round as input pairs for the next round. One may combine $m$ of these basic purification steps, thereby obtaining a protocol that operates on $2^m$ input pairs which produces (probabilistically) one output pair. The corresponding resource state for a measurement-based implementation at each site is of size $2^m+1$. This setting has been studied for $m \leq 2$ in \cite{ZwergerRepeater}. Notice that the resource state for performing several purification rounds in one steps contains fewer qubits than the resource states to perform the same task in a sequential fashion. For two rounds, one needs at each site three three-qubit states, i.e. nine qubits, when performing the protocol in a sequential fashion, while a resource state of five qubits (four input plus one output) suffices for the overall task. This reduction in size of resource states seems to be the crucial feature that leads to very high error thresholds in a measurement-based implementation, where a threshold of more than $23\%$ noise per qubit was found \cite{ZwergerEPP}. In turn, performing multiple rounds in one step leads to a smaller success probability. Obtaining the explicit form of resource states that allow one to perform entanglement purification with such a high robustness and tolerance against noise and imperfections is highly relevant, and will be the subject of the remainder of this section.

We emphasize that our results go beyond the analysis provided in \cite{ZwergerEPP,ZwergerRepeater}, as we offer a framework for constructing analytically the \emph{concrete} stabilizers of the resource state for an \emph{arbitrary} number of rounds of entanglement purification explicitly rather than considering only a small number of entanglement purification rounds. Fig. \ref{fig.epp.meas} shows the measurement-based implementation of two rounds of the DEJMPS protocol.
\begin{figure}[h!]
\scalebox{0.8}{
\includegraphics{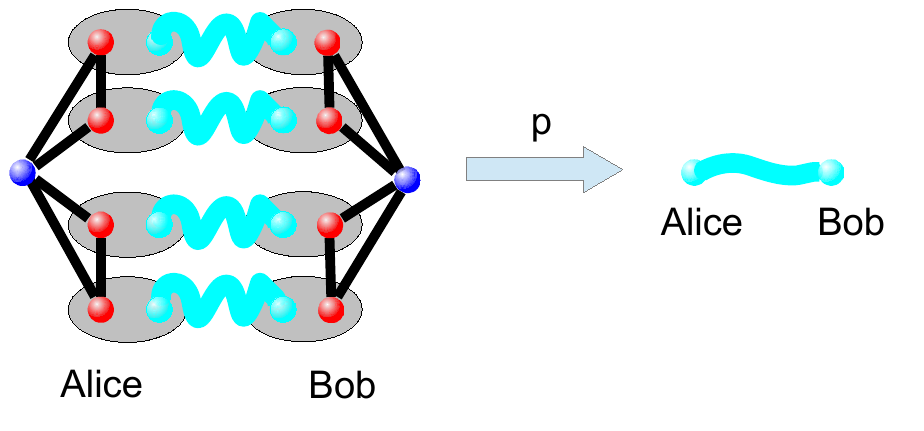}
}
\caption[h!]{Measurement-based implementation of two rounds of the DEJMPS entanglement distillation protocol. The protocol operates on four input pairs and produces one output pair probabilistically.} \label{fig.epp.meas}
\end{figure}

\subsubsection{Construction of resource state for several rounds of the DEJMPS protocol}

The resource state for one basic purification step of the DEJMPS protocol at Alice's side \cite{ZwergerRepeater} is
\begin{align}
\ket{\psi_{\text{D}}} = 1 / \sqrt{2} \left(\ket{-} \ket{\phim} - i \ket{+} \ket{\psip} \right).
\end{align}
Hence we have that $\ket{G_0} = -i \ket{\psip}$ and $\ket{G_1} = \ket{\phim}$. The resource state at Bobs side differs in the sign of $\ket{G_0}$. \newline
One easily verifies that the initial auxiliary operators of $\ket{\psi_{\text{D}}}$ are $K_1 = -(Y \otimes \id)$, $K'_1 = -(\id \otimes Y)$ and $F_1 = -(Z \otimes Z)$. Hence $\K_1 = \lbrace -(Y \otimes \id), -(\id \otimes Y) \rbrace$ and $\F_1 = \lbrace -(Z \otimes Z) \rbrace$. Thus we obtain via Theorem \ref{thm.koperators} that $\gamma_1=\delta_1=(1,0)$ and $\gamma_2=\delta_2=(0,1)$ where $c = -i$ for both cases. Furthermore Theorem \ref{thm.foperators1} implies $\epsilon=(1,1)$, $\eta=(0,0)$ and $c=-1$. To summarize, the sets of auxiliary operators $\K_{n+1}$ and $\F_{n+1}$ are given by
\begin{align}
\K_{n+1} &= \lbrace -i(F_{n} K_{n} \otimes \id): \, F_{n} \in \F_n, K_{n} \in \K_{n} \rbrace \label{dejmps:krec1} \\
& \cup \lbrace -i(\id \otimes F_{n} K_{n}): \, F_{n} \in \F_n, K_{n} \in \K_{n} \rbrace \label{dejmps:krec2}, \\
\F_{n+1} &= \lbrace -(K_n \otimes K_n): \, F_{n} \in \F_n, K_{n} \in \K_{n} \rbrace \label{dejmps:frec}.
\end{align}
The stabilizers of the resource state for performing $n$ purification steps of the DEJMPS protocol are, according to (\ref{framework:kstab}) and (\ref{framework:fstab}), $\lbrace Z \otimes K_{n} \rbrace_{K_n \in \K_n}$ and $\lbrace X \otimes F_{n} \rbrace_{F_n \in \F_n}$. One easily verifies that those stabilizers are linear independent for all $n$ by construction.

\subsubsection{Graph state representation of resource states}

In the following we transform the resource states of the DEJMPS protocol via local unitaries to graph states and show that there exists a construction scheme solely based  on simple graph rules for concatenating the DEJMPS protocol. For that purpose we observe from (\ref{dejmps:krec2}) and (\ref{dejmps:frec}) that
\begin{align}
F_{n-1} K_{n-1} &= -(K_{n-2} \otimes K_{n-2})(-i)(\id \otimes F_{n-2} K_{n-2}) \notag \\
& = i (K_{n-2} \otimes K_{n-2} F_{n-2} K_{n-2}) \notag \\
& = i \delta_i (K_{n-2} \otimes F_{n-2}) \label{eq:dejmps:arb:1}
\end{align}
for $K_{n-2} \in \K_{n-2}$ and $F_{n-2} \in \F_{n-2}$ where $\delta_i \in \lbrace 1, -1 \rbrace$ depends on whether $K_{n-2}$ and $F_{n-2}$ commute or anticommute. Similarly one obtains via (\ref{dejmps:krec1}) and (\ref{dejmps:frec}) that
\begin{align}
F_{n-1} K_{n-1} =  i \delta_i (F_{n-2} \otimes K_{n-2}) \label{eq:dejmps:arb:2}
\end{align}
for $K_{n-2} \in \K_{n-2}$ and $F_{n-2} \in \F_{n-2}$. Inserting (\ref{eq:dejmps:arb:1}) and (\ref{eq:dejmps:arb:2}) in (\ref{dejmps:krec1}) and (\ref{dejmps:krec2}) implies that
\begin{align}
\K_{n} &= \lbrace \delta^{(1)}_{KF} \id \otimes \id \otimes K \otimes F: K \in \K_{n-2}, F \in \F_{n-2} \rbrace \notag \\
& \cup \lbrace \delta^{(2)}_{KF} \id \otimes \id \otimes F \otimes K : K \in \K_{n-2}, F \in \F_{n-2} \rbrace \notag \\
& \cup \lbrace \delta^{(3)}_{KF} K \otimes F \otimes \id \otimes \id: K \in \K_{n-2}, F \in \F_{n-2} \rbrace \notag \\
& \cup \lbrace \delta^{(4)}_{KF} F \otimes K  \otimes \id \otimes \id : K \in \K_{n-2}, F \in \F_{n-2} \rbrace
\end{align}
where $\delta^{(i)}_{KF} \in \lbrace -1, 1 \rbrace$ for $1 \leq i \leq 4$ depend on $K \in \K_{n-2}$ and $F \in \F_{n-2}$. One computes in a similar fashion for $F_{n}$ that $F_{n} = -(K_{n-1} \otimes K_{n-1}) = \id \otimes F_{n-2} K_{n-2} \otimes \id \otimes F_{n-2} K_{n-2}$ for $K_{n-2} \in \K_{n-2}$ and $F_{n-2} \in \F_{n-2}$. To summarize, the stabilizers of the resource state for $n$ rounds of purification are
\begin{align}
\delta_{K'_1 F'_1} & Z \otimes \id \otimes \id \otimes K'_1 \otimes F'_1 \label{eq:dejmps:arb:stab1}\\
\delta_{K'_2 F'_2} & Z \otimes \id \otimes \id \otimes F'_2 \otimes K'_2 \label{eq:dejmps:arb:stab2}\\
\delta_{K'_3 F'_3} & Z \otimes K'_3 \otimes F'_3 \otimes \id \otimes \id \label{eq:dejmps:arb:stab3}\\
\delta_{K'_4 F'_4} & Z \otimes F'_4 \otimes K'_4  \otimes \id \otimes \id \label{eq:dejmps:arb:stab4}\\
& X \otimes \id \otimes F'_5 K'_5 \otimes \id \otimes F'_6 K'_6 \label{eq:dejmps:arb:stab5}
\end{align}
where $K'_{i} \in \K_{n-2}$ and $F'_{i} \in \F_{n-2}$ for all $1 \leq i \leq 6$.  Multiplying (\ref{eq:dejmps:arb:stab2}), (\ref{eq:dejmps:arb:stab4}) and (\ref{eq:dejmps:arb:stab5}) with appropriate choices of $K'_2=K'_6$ and $K'_4=K'_5$,  yields the stabilizer $\delta' X \otimes F'_{4} \otimes F'_{5} \otimes F'_{2} \otimes F'_{6}$, where we have used that Pauli operators either commute or anti-commute. Thus we obtain as stabilizers of the resource state for $n+2$ rounds of the DEJMPS
\begin{align}
\delta_{K'_1 F'_1} & Z \otimes \id \otimes \id \otimes K'_{1} \otimes F'_{1} \label{eq:dejmps:arb:stab:new1} \\
 \delta_{K'_2 F'_2} & Z \otimes \id \otimes \id \otimes F'_{2} \otimes  K'_{2} \label{eq:dejmps:arb:stab:new2} \\
 \delta_{K'_3 F'_3} & Z \otimes K'_{3} \otimes F'_{3} \otimes \id \otimes  \id \label{eq:dejmps:arb:stab:new3} \\
 \delta_{K'_4 F'_4} & Z \otimes F'_{4} \otimes K'_{4} \otimes \id \otimes  \id \label{eq:dejmps:arb:stab:new4} \\
 \delta' & X \otimes F'_{5} \otimes F'_{6} \otimes  F'_{7} \otimes F'_{8} \label{eq:dejmps:arb:stab:new5}
\end{align}
where $K'_{i} \in \K_{n-2}$ for $1 \leq i \leq 4$ and $F'_{j} \in \F_{n-2}$ for all $1 \leq j \leq 8$. From $F_{1} = Z^{\otimes 2}$ and $F_{2} = Z^{\otimes 4}$ we deduce
\begin{align}
F_{n} = Z^{\otimes 2^n}, \label{eq:dejmps:arb:fallz}
\end{align}
which implies $|\F_{n}| = 1$ for all $n$. Now we show that the stabilizers (\ref{eq:dejmps:arb:stab:new1})-(\ref{eq:dejmps:arb:stab:new5}) correspond to a graph state up to local Clifford operations. \newline
We observe from (\ref{eq:dejmps:arb:stab1})-(\ref{eq:dejmps:arb:stab4}) that if the operators in $\K_{1}$ (odd number of purification rounds) or $\K_{2}$ (even number of purification rounds) contain exactly \emph{one} $Y$ or $X$ respectively then all subsequent $K_{n} \in \K_{n}$ will also contain exactly \emph{one} $Y$ or $X$. (\ref{eq:dejmps:arb:fallz}) implies that the output particle will be connected to all input particles and that the stabilizers describe a valid graph state. \newline
Furthermore we find from (\ref{eq:dejmps:arb:fallz}) and (\ref{eq:dejmps:arb:stab:new1}) - (\ref{eq:dejmps:arb:stab:new4}) that the resource state for $n$ rounds of purification is a tensor product corresponding to two disjoint subgraphs, each described by (\ref{eq:dejmps:arb:stab:new1})-(\ref{eq:dejmps:arb:stab:new2}) and (\ref{eq:dejmps:arb:stab:new3})-(\ref{eq:dejmps:arb:stab:new4}). Both subgraphs are of the form
\begin{align}
 \delta_{K'_1 F'_{n-2}} & Z \otimes K'_{1} \otimes F_{n-2}, \label{eq:dejmps:arb:scaling:k:1}\\
\delta_{K'_2 F'_{n-2}} & Z \otimes F_{n-2} \otimes K'_{2} \label{eq:dejmps:arb:scaling:k:2}
\end{align}
where $K'_{i} \in \K_{n-2}$ for $1 \leq i \leq 2$. From (\ref{eq:dejmps:arb:scaling:k:1}) and (\ref{eq:dejmps:arb:scaling:k:2}) we infer the following construction scheme for the resource state of the DEJMPS protocol:
\begin{enumerate}
	\item Remove the output qubit from the resource state. This results in two disjoint graphs $G_1$ and $G_2$. \label{enu:dejmps:arb:const:1}
	\item Duplicate the disjoint graphs of item \ref{enu:dejmps:arb:const:1} and label them with $G_3$ and $G_4$. \label{enu:dejmps:arb:const:2}
	\item Connect $G_1, G_2, G_3$ and $G_4$ according to the following rules:
	\begin{itemize}
		\item Each particle of $G_1$ with each particle of $G_3$
		\item Each particle of $G_1$ with each particle of $G_4$
		\item Each particle of $G_2$ with each particle of $G_3$
		\item Each particle of $G_2$ with each particle of $G_4$
	\end{itemize}
	We denote the resulting graph $G'_1$.
	\label{enu:dejmps:arb:const:3}
	\item Duplicate the graph $G'_1$ of \ref{enu:dejmps:arb:const:3} and denote it $G'_2$. \label{enu:dejmps:arb:const:4}
	\item Insert the output qubit and connect it to all particles in $G'_1$ and $G'_2$.
	\label{enu:dejmps:arb:const:5}
\end{enumerate}

The construction scheme is depicted in Fig. \ref{fig:1epp-3epp} and \ref{fig:3epp-5epp} for the resource state of three and five rounds of entanglement purification.
The resulting resource states for one, two and three steps of the DEJMPS protcol are depicted in Fig. \ref{fig:123epp}.

\begin{figure}[h!]
\scalebox{0.6}{
\includegraphics{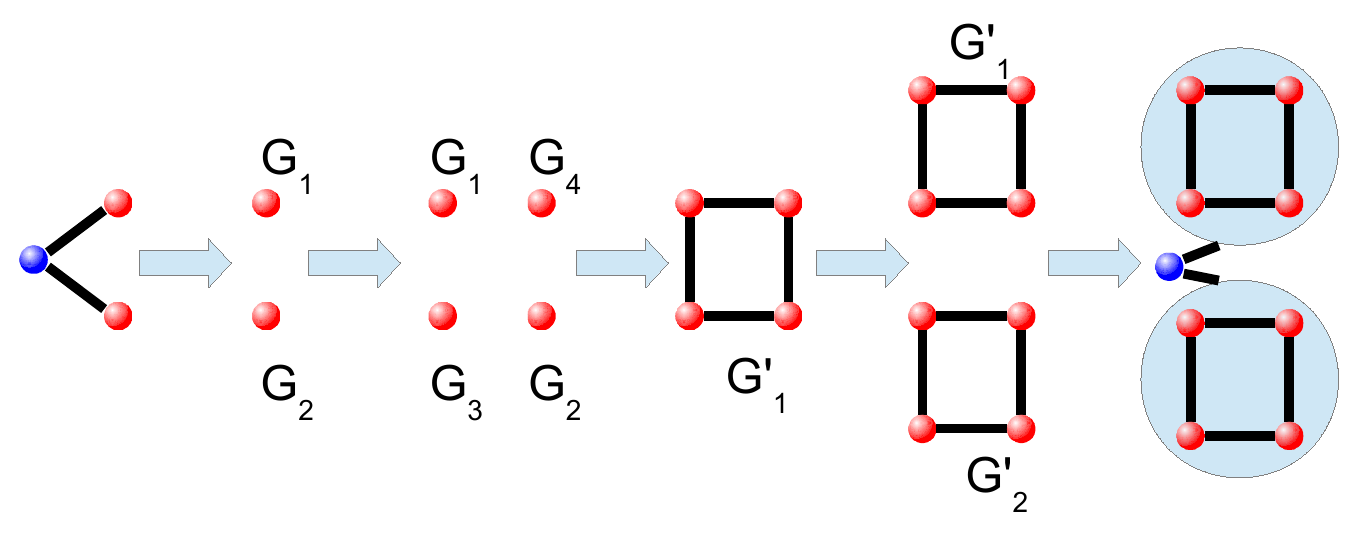}
}
\caption[h!]{Constructing the resource state for three rounds of entanglement purification from the resource state for one purification round.}\label{fig:1epp-3epp}
\end{figure}
\begin{figure}[h!]
\scalebox{0.6}{
\includegraphics{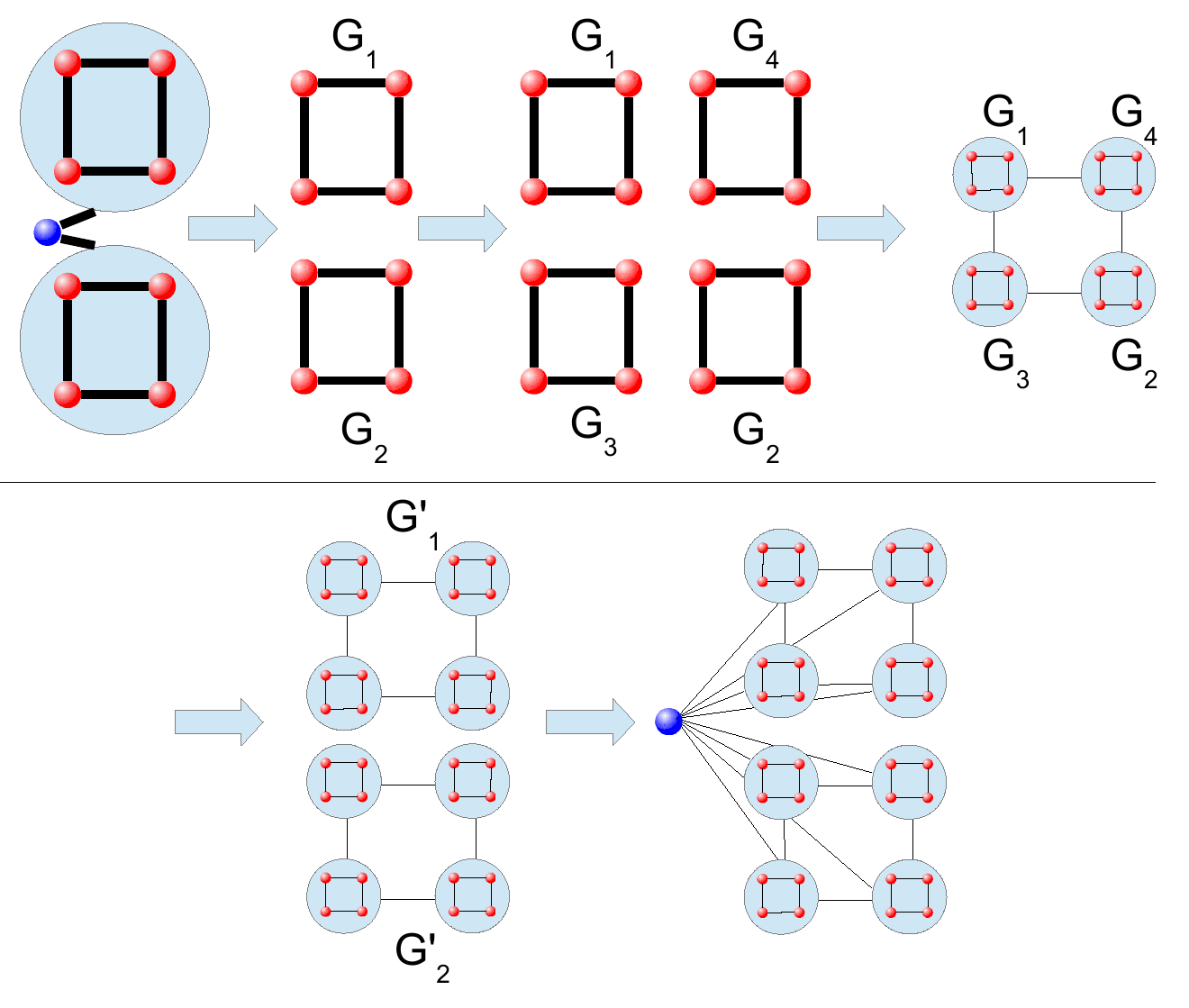}
}
\caption[h!]{Constructing the resource state for five rounds of entanglement purification from the resource state for three purification rounds.} \label{fig:3epp-5epp}
\end{figure}

\begin{figure}[h!]
\scalebox{0.85}{
\includegraphics{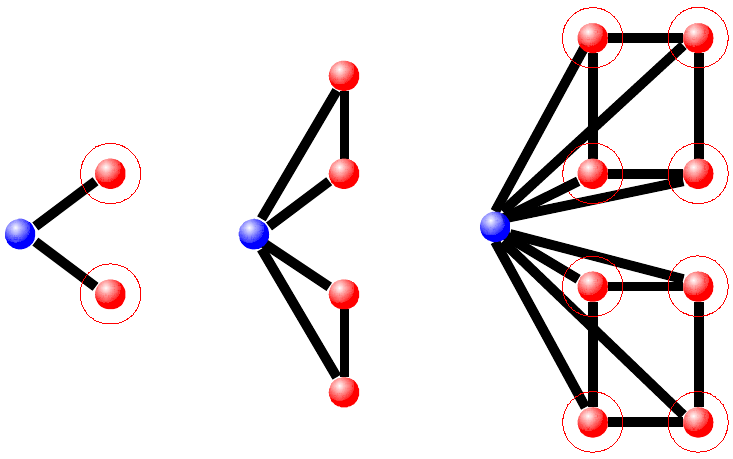}
}
\caption[h!]{Graph state representation of resource states for one, two and three rounds of entanglement purification using the DEJMPS protocol. The red circle indicates a $R_z(-\pi/2)$ rotation.}\label{fig:123epp}
\end{figure}

Since all quantum gates involved in the DEJMPS protocol are elements of the Clifford group and only measurements of the $Z$ observable are performed, we can propagate possible errors due to the read-in Bell-measurement through the resource state, see Sec. \ref{sec:back:stab}. A detailed analysis of the resulting Pauli corrections necessary at the output qubit is provided in \cite{ZwergerRepeater} for small resource states, but it is straightforward to explicitly obtain a similar result for multiple rounds of entanglement purification.

\subsection{Quantum repeaters}

Quantum repeaters enable long-distance quantum communication by combining entanglement purification and entanglement swapping \cite{BriegelRepeater}. This allows one to establish a long-distance entangled pair over arbitrary distances, with an overhead that scales only polynomially (in terms of resources, success probability, or time) with the distance. To this aim, the channel is divided into segments of shorter distance, and (noisy) short-distance entangled pairs are generated over all segments. Entanglement purification is used to generate pairs with high fidelity, after which two neighboring pairs are connected by means of Bell measurement. The measurement within the Bell-basis and the classical communication of the outcome is also referred to as entanglement swapping. A repeater station thus distills from several noisy Bell-pairs one Bell-pair of sufficiently high fidelity relative to $\ket{\phip}$, both for pairs from left and right. The resulting output qubits (more precisely, one part of a Bell-pair) belonging to different segments get measured within the Bell-basis and the obtained outcome is classically communicated. This establishes a long-distance quantum communication link by means of quantum teleportation. Notice that at each repeater station, there are only input and no output particles, i.e. no particles are left after the protocol has finished. Only at the end stations at Alice and Bob, no entanglement swapping is performed, and an output qubit is kept.
The basic functionality of a quantum repeater is shown in Figure \ref{fig:repeater}.

\begin{figure}[h!]
\scalebox{1}{
\includegraphics{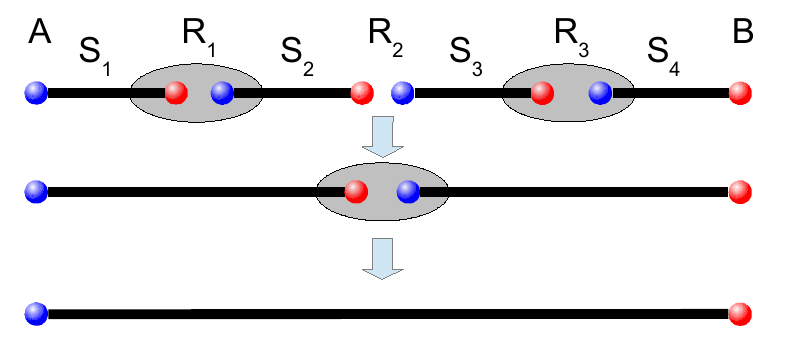}
}
\caption[h!]{Quantum repeaters establish long-distance quantum communication by combining entanglement purification and entanglement swapping.}\label{fig:repeater}
\end{figure}

If the quantum repeater station relies on a measurement-based implementation \cite{ZwergerRepeater} then the resource state of a quantum repeater is the combination of the quantum tasks \emph{entanglement purification} and \emph{entanglement swapping}. Hence we construct the explicit resource state via connecting the resource states for entanglement purification of two segments via a Bell-measurement, see Figure \ref{fig:thm:coupling}.


Now we translate this construction into the proposed framework. Because entanglement swapping coincides with a Bell-measurement if the outcome is $\ket{\phip}$, we easily find the stabilizers of the resource state at a quantum repeater via Theorem \ref{thm.coupling}. Hence the resource state is stabilized by the operators
\begin{align}
\lbrace K_{S_i} & \otimes K_{S_{i+1}} : K_{S_i} \in \K_{S_{i}}, K_{S_{i+1}} \in \K_{S_{i+1}} \rbrace \notag \\
& \cup \lbrace F_{S_i} \otimes F_{S_{i+1}} : F_{S_i} \in \F_{S_{i}}, F_{S_{i+1}} \in \F_{S_{i+1}} \rbrace \label{eq:swap:stab}
\end{align}
where $\K_{S_i}, \F_{S_i}$ and $\K_{S_{i+1}}, \F_{S_{i+1}}$ denote the sets of auxiliary operators of the entanglement purification resource states for segments $S_{i}$ and $S_{i+1}$. Using (\ref{eq:swap:stab}) one can apply Theorem 1 of \cite{Nest} to obtain a Clifford-equivalent graph state. In \cite{ZwergerRepeater} the resource states for one and two purification rounds followed by entanglement swapping were constructed explicitly, whereas here, with the proposed scheme, the stabilizers of resource states for \emph{arbitrary} rounds of entanglement purification followed by entanglement swapping can easily be constructed explicitly. 
%

\subsection{Quantum error correction}

Quantum error correction codes are used to protect quantum information against different types of errors \cite{Nielsen}. The basic idea is to encode a physical qubit as a logical qubit in a higher dimensional Hilbert space in such a way that any error maps the two-dimensional code space to an orthogonal two-dimensional subspace. By identifying the resulting error subspace, one can detect and correct the corresponding error. It is thereby sufficient to consider only Pauli errors. There exist codes that protect only against bit flip or phase flip errors, which are simple variants of classical repetition codes. One can combine these codes to obtain an error correction code that protects a single logical qubit against arbitrary errors on one of the physical qubits. This leads to the Shor code \cite{ShorCode}, an error correction code that uses nine qubits to encode a logical qubit. The optimal code to protect a single logical qubit against arbitrary errors on one of the encoding qubits is of size five, where e.g. a cluster-ring code can be used \cite{Grassl}. A multitude of good quantum error correction codes are known that allow one to protect a single (or several) qubits against arbitrary errors occurring on a certain limited number of physical qubits, as long as the error probability is small enough. One possible way to construct such large-scale codes is by concatenating small-scale codes \cite{Nielsen,gottesmanphd}, where the Shor code is the basic example. The generalized Shor code makes use of this principle, and leads in fact to a code with a very high error threshold against individual noise \cite{GenShorCode1,GenShorCode2,GenShorCode3,GenShorCode4}. The concatenation of quantum error correction codes found promising applications \cite{Kesting,LanyonQEC,DawsonNoiseQuantumComputers,DawsonNoiseClusterState}. This issue has been investigated in terms of stabilizers also in \cite{gottesmanphd} with a similar result as here, but has not been applied to measurement-based implementations. Quantum error correction codes which relate codewords to graph states have been studied in detail in  \cite{SchlingemannGraphCode,Grassl,HeinGraph,HeinMulti}. Other approaches include topological codes \cite{ToricCode,GraphCodes}. All these examples have in common that they are stabilizer codes \cite{Gottesman}, i.e. the codewords $|0_L\rangle$ and $|1_L\rangle$ are stabilizer states. For all these codes, encoding, decoding and syndrome readout can be done using only Clifford circuits and Pauli measurements, which allows for an efficient measurement-based implementation \cite{ZwergerQEC,ZwergerQC}. Elements of measurement-based quantum error correction have in fact already been experimentally demonstrated using trapped ions \cite{LanyonQEC} and photons \cite{BarzQEC}. Furthermore, \cite{DawsonNoiseQuantumComputers,DawsonNoiseClusterState} numerically found that optical cluster-state quantum computation \cite{BrownerOptical,NielsenOptical} using quantum error correction codes is practical, which enables scalable optical quantum computation. These results motivate the following subsection.

\subsubsection{Encoding}
In this section, we consider such a measurement-based approach to quantum error correction. We first concentrate on the encoding operation, as the resource states for decoding is given by the same state, while the resource state for syndrome measurement (error correction step) can be obtained by combining the resource states for decoding and encoding.  For encoding, one needs to perform the mapping $\ket{0} \mapsto \ket{0_L}$ and $\ket{1} \mapsto \ket{1_L}$ where $\ket{0_L}$ and $\ket{1_L}$ denote the codewords of the corresponding error correction code.
A measurement-based implementation is done by preparing a resource state and executing Bell-measurements at the read-in. One easily verifies that the resource state for encoding and decoding in a particular quantum error correction code is given by
\begin{align}
\ket{\psi} &= (\ket{0}\ket{0_L} + \ket{1}\ket{1_L})/\sqrt{2} \label{qec:general} \\
&= (\ket{+}(\ket{0_L} + \ket{1_L}) + \ket{-}(\ket{0_L} - \ket{1_L}))/2. \notag
\end{align}
Therefore we have according to (\ref{eq:observation:plusminus}) that
\begin{align}
\ket{G_0} &= (\ket{0_L} + \ket{1_L})/2 \label{qec:g0} \\
\ket{G_1} &= (\ket{0_L} - \ket{1_L})/2 \label{qec:g1}.
\end{align}
As the read-in Bell-measurement for encoding quantum information is probabilistic we need to take care of different outcomes. We handle them as follows: Suppose we encode $\ket{\varphi} = \alpha \ket{0} + \beta \ket{1}$ using the resource state $\ket{\psi}$ of (\ref{qec:general}). Furthermore assume that the sets of auxiliary operators $\K$ and $\F$ of $\ket{\psi}$ are known. A straightforward computation yields for the state after the read-in (depending on the measurement outcome)
\begin{align}
\alpha \ket{0_L} + \beta \ket{1_L} \text{ for the $\ket{\phip}$-outcome}, \\
\beta \ket{0_L} + \alpha \ket{1_L} \text{ for the $\ket{\psip}$-outcome}, \\
\alpha \ket{0_L} - \beta \ket{1_L} \text{ for the $\ket{\phim}$-outcome}, \\
\beta \ket{0_L} - \alpha \ket{1_L} \text{ for the $\ket{\psim}$-outcome}.
\end{align}
A simple computation shows $K \ket{0_L} = K (\ket{G_0} + \ket{G_1}) = \ket{0_L}$ and $K \ket{1_L} = K (\ket{G_0} - \ket{G_1}) = -\ket{1_L}$ for $K \in \K$. Similarly one  obtains $F \ket{0_L} = \ket{1_L}$ and $F \ket{1_L} = \ket{0_L}$ for $F \in \F$. Hence operators in $\K$ are logical $Z$ operators whereas operators in $\F$ are logical $X$ operators. Table \ref{tab:qec:byenc:corr} enables us to determine the correction operator for different measurement outcomes at the read-in on encoding.\newline
\begin{center}
\begin{table}[h!]
\begin{tabular}{ c | c }
Outcome & Correction \\
\hline
$\ket{\phip}$ & $\id$ \\
$\ket{\psip}$ & $F$ \\
$\ket{\phim}$ & $K$ \\
$\ket{\psim}$ & $FK$
\end{tabular}
\caption{Correction operators for different measurement outcomes at the read-in where $K \in \K$ and $F \in \F$.} \label{tab:qec:byenc:corr}
\end{table}
\end{center}

\subsubsection{Decoding and correction}
Notice that the same resource state as for encoding can be used for decoding, where now the role of input and output qubits is exchanged. The decoding procedure in fact also involves an error correction step, where the required correction operation is determined by the outcomes of the in-coupling Bell measurements \cite{ZwergerQC,ZwergerQEC,LanyonQEC,BarzQEC}.
Determining the decoding correction operation is a bit more complex than in the case of encoding, as there are more possible combinations of measurement outcomes to analyse. While it is of course possible to obtain the corrections from the underlying construction of the resource state via the Jamiolkowski isomorphism, they can also be determined directly from the stabilizers and the auxiliary operators $K$ and $F$.\newline
The four Bell states can be written as $(\id \otimes \sigma_{i,j}) \ket{\phip}$ with $i=0,1$ and $j=0,1$ so the projection on any combination of Bell states can be understood as applying a combination of Pauli operators on the input state $\ket{\Psi_{\mathrm{in}}},$ followed by projections on $\ket{\phip}$. A particular combination of different measurement outcomes can only occur if the overlap with $\ket{\phip}^{\otimes n}$ is non-zero after applying these Pauli operators. This is the case only if $\sigma_{i_1,j_1} \otimes \dots \otimes \sigma_{i_n,j_n} \ket{\Psi_{\mathrm{in}}}$ is in the logical subspace. \newline
That means if no error occured only measurement outcomes corresponding to $\sigma_{i_1,j_1} \dots \sigma_{i_n,j_n}$ which leave the logic subspace invariant can occur. This is precisely the logical group $S_L$ of the underlying error correction code, which can be obtained from the auxiliary operators $K$ and $F$ in a straightforward way. Some combinations of Pauli operators act as the identity on the logic subspace and those can be found by going over all different ways to multiply two $K$ operators or two $F$ operators together. The representation of the logical operator $X_L$ [$Y_L$,$Z_L$] is not unique either and the different representation can be found by choosing one representation and multiplying it by the different representions of the identity on the logic subspace. \newline
If the outcome corresponds to a combination of Pauli operators that is one representation of $X_L$, then the correction operation on the output is $X$. This can be understood in a similar fashion to the standard quantum teleportation with one party holding a logical qubit instead of a single physical one. Therefore the measurement outcome corresponding to $X_L$ is a projection on $\ket{\psip},$ but on the logical level. The same argument holds for outcomes corresponding to $Y_L$ and $Z_L$ which lead to corrections $Y$ and $Z$ respectively. \newline
Any other combination of measurement outcomes corresponds to subspaces where an error is detected. If a code can correct the Pauli errors $E_k$, the set $E_k S_L$ corresponds to detecting the error $E_k$. Similar as before, the correction operation is given by $X$ [$Y$,$Z$] if the measurement outcomes correspond to $E_k X_L$ [$E_k Y_L$,$E_k Z_L$]. If the error correction code is not optimal, for the remaining combinations one needs to define a default correction as these belong to the subspaces that correspond to errors the code can detect, but not correct. \newline
This approach is also applicable to resource states which implement quantum circuits other than error correction. The role of $K$ and $F$ stay the same, even though it does not make sense to speak of logical operators and subspaces in that case.

In the following we first provide results concerning the bit- and phase-flip code for pedagogical reasons and generalize this results to a generalized Shor code \cite{GenShorCode1,GenShorCode2,GenShorCode3,GenShorCode4}. We have also analysed the five qubit cluster-ring code within the proposed framework and refer the interested reader to Appendix \ref{app:cluster}. We emphasize that the framework is especially suited to construct resource states for  concatenated quantum error correction codes which offer promising error thresholds \cite{Kesting,GenShorCode1,GenShorCode2,GenShorCode3,GenShorCode4}.

\subsubsection{Bit-flip code}

For illustration purpose we first discuss the bit-flip code. The $m-$qubit bit-flip code protects a logical qubit against up to $(m-1)/2$ bit flip errors, where we assumed that $m$ is odd. Therefore we have the encoding $\ket{0_L} = \ket{0}^{\otimes m}$ and $\ket{1_L} = \ket{1}^{\otimes m}$. One easily verifies via (\ref{qec:g0}) and (\ref{qec:g1}) that $\K_1 = \lbrace Z^{(i)}: 1 \leq i \leq m \rbrace$ and $F_1 = X^{\otimes m}$. Hence, via Theorem \ref{thm.koperators} we easily obtain $\gamma_i = (0..,0,1,0,..,0)$ where only the $i-$th entry of $\gamma_i$ is $1$ for $i=1,..,m$ as well as $\delta=(0,..,0)$, thus $\K_{n+1} = \lbrace K^{(i)}: 1 \leq i \leq m, K \in \K_n \rbrace$. For the operators $F_n$ we apply Theorem \ref{thm.foperators1} and find that $\epsilon = (0,..,0)$, $\eta=(1,..,1)$ and $c = 1$. Hence Theorem \ref{thm.foperators1} implies that $\F_{n+1} = \lbrace \bigotimes^m_{j = 1} F_{j}: F_j \in \F_n \rbrace$. Thus the resource state of the bit-flip code is transformed into a graph state by applying a Hadamard gate on all output particles leading to a $m^n + 1$ GHZ state. Furthermore the obtained stabilizers are linear independent by construction.


\subsubsection{Phase-flip code}

The $m-$qubit phase-flip code is used to correct phase-flip errors affecting the encoded state, and can be obtained from the bit-flip code by applying Hadamard operations. Thus, the logical zero and logical one are given by $\ket{0_L} = \ket{+}^{\otimes m} $ and $\ket{1_L} = \ket{-}^{\otimes m}$ respectively. Hence the initial auxiliary operators are given by $\K_1 = \lbrace X^{(i)}: 1 \leq i \leq m \rbrace$ where  and $F_1 = Z^{\otimes m}$ respectively. Theorem \ref{thm.koperators} implies that $\gamma=(0,..,0)$ and $\delta_i=(0,..,0,1,0,..,0)$ where only the $i-th$ entry of $\delta_i$ is $1$ for $i=1,..,m$, thus $\K_{n+1} = \lbrace F^{(i)}: 1 \leq i \leq m, F \in \F_n \rbrace$. Finally, Theorem \ref{thm.foperators1} gives $\epsilon=(1,..,1)$ and $\eta=(0,0,0)$, hence $\F_{n+1} = \lbrace \bigotimes^m_{j = 1} K_{j}: K_j \in \K_n \rbrace$.

\subsubsection{Shor code and generalized Shor code}

The Shor code \cite{ShorCode} consists of a concatenation of phase-flip and bit-flip code, each of size three. This leads to a nine-qubit code that is capable of protecting quantum information against one arbitrary error happening on one physical qubit. The idea of the Shor code can be generalized to arbitrary numbers of bit-flip and phase-flip encodings, see \cite{GenShorCode1,GenShorCode2,GenShorCode3,GenShorCode4,Kesting}. The resulting quantum error correction code is constructed as follows: First encode the qubit in an $m_1-$qubit phase-flip code followed by a $m_2-$qubit bit-flip code.
The codewords are thus given by
\begin{align}
\ket{0_L} &= \left(\frac{\ket{0}^{\otimes m_2} + \ket{1}^{\otimes  m_2}}{\sqrt{2}}\right)^{\otimes m_1} \\
\ket{1_L} &= \left(\frac{\ket{0}^{\otimes m_2} - \ket{1}^{\otimes m_2}}{\sqrt{2}}\right)^{\otimes m_1}.
\end{align}
The resource state necessary to implement this quantum error correction code is given by
\begin{align}
\ket{\psi} = & \ket{0} \left(\frac{\ket{0}^{\otimes m_2} + \ket{1}^{\otimes  m_2}}{\sqrt{2}}\right)^{\otimes m_1} / \sqrt{2} \notag \\
& + \ket{1} \left(\frac{\ket{0}^{\otimes m_2} - \ket{1}^{\otimes m_2}}{\sqrt{2}}\right)^{\otimes m_1} / \sqrt{2}
\end{align}
up to a global normalization factor. We call such a code a $[m_1, m_2]-$generalized Shor code. Now we will construct the stabilizers of $\ket{\psi}$ in terms of auxiliary operators. For that purpose, we observe that we combine the resource state of an $m_1-$qubit phase-flip code with the resource state of an $m_2-$qubit bit-flip code. Hence by inserting the initial auxiliary operators of an $m_2-$qubit bit-flip code, i.e. $\K_1 = \lbrace Z^{(i)} :  1 \leq i \leq m_2 \rbrace$ and $F_1 = X^{\otimes m_2}$ respectively, in the recurrence relations of an $m_1-$qubit phase-flip code $\K_{n+1} = \lbrace F^{(i)} : 1 \leq i \leq m_1 \, , F \in \F_n \rbrace$ and $\F_{n+1} = \lbrace \bigotimes^m_{j = 1} K_{j}: K_j \in \K_n \rbrace$ we find with the help of (\ref{framework:kstab}) and (\ref{framework:fstab}) the stabilizers for a $[m_1,m_2]-$generalized Shor code to be
\begin{align}
\lbrace Z & \otimes (X^{\otimes {m_2}})^{(j)}: 1 \leq j \leq m_1 \rbrace \label{qec:genshor:stab:1}\\
& \cup \left\lbrace X \otimes \bigotimes\limits^{m_1 - 1}_{k = 0} Z^{(i_k + {m_2} k)}: 1 \leq i_k \leq m_2 \right\rbrace. \label{qec:genshor:stab:2}
\end{align}
The stabilizers (\ref{qec:genshor:stab:1}) and (\ref{qec:genshor:stab:2}) can be transformed by applying local unitaries to graph state stabilizers, see Appendix \ref{app:genshor} for details. The result thereof is depicted in Fig. \ref{fig:genshor}.

\begin{figure}[h!]
\scalebox{0.5}{
\includegraphics{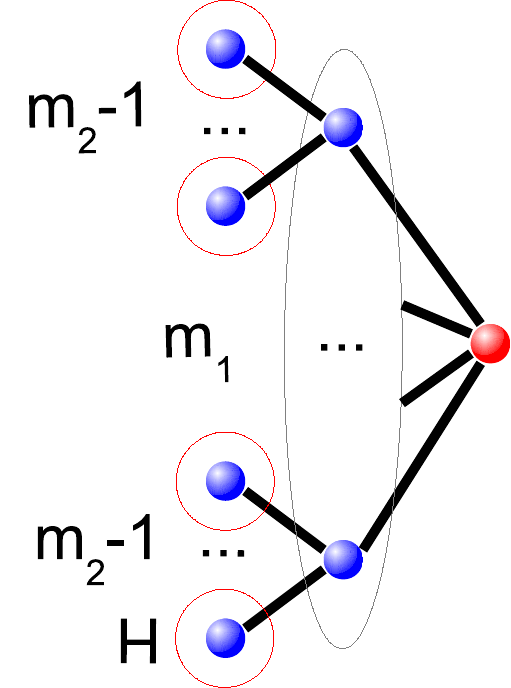}
}
\caption[h!]{Resource state of a $[m_1,m_2]-$generalized Shor code. The red circles indicate that a Hadamard rotation needs to be applied at the end. The red vertex corresponds to the input qubit, the blue vertices to the output qubits.}\label{fig:genshor}
\end{figure}

\subsubsection{Error syndrome readout and Code switcher}

The resource state for error correction (error syndrome readout and correction) can be obtained from the resource states from decoding and encoding. The output qubit of the decoding resource state is connected via a Bell measurement to the input qubit of the resource state for encoding, where the initial and the final error correction codes are the same. We emphasize that this is only used as a tool to construct the resource state for the overall process, and quantum information is at no place actually decoded. Quantum information remains protected in the quantum error correction code at all times, and the required correction operation is determined from the results of the in-coupling Bell measurements.

Notice that one may use different error correction codes for decoding and encoding. In this way, one obtains a resource state for a {\em code switcher} (with build-in error correction) that allows one to switch between two arbitrary stabilizer codes. Such a code switcher is a useful tool in several contexts. For instance, it can be used to switch between an error correction code for storage of quantum information, and a code for processing the information \cite{NautrupTopo}. For the former, a high stability against noise and decoherence is crucial, where also a passive or topological protection may be applicable. The requirements on a code for processing information, e.g. in a fault-tolerant quantum computation architecture, are different and also include the simple realization of certain encoded gates. Codes that are transversal for certain kinds of operations are known, but not all operations can be performed transversally using the same code. We remark that one can easily construct resource states to realize logical Clifford operations for all Calderbank-Shor-Steane or stabilizer codes \cite{ZwergerQEC} using only input and output qubits, and with build-in error correction. Here we discuss how to obtain a general code switcher (including error correction), which includes error syndrome readout for a fixed code as special case.

In the proposed framework, this boils down to combining a decoding resource state and an encoding resource state via a Bell-measurement, see Fig. \ref{fig:thm:coupling}. Thus we apply Theorem \ref{thm.coupling} which implies that the stabilizers of a code switcher, by denoting the sets of auxiliary operators of encoding and decoding resource state by $\K_{\text{E}}$ and $\F_{\text{E}}$ as well as $\K_{\text{D}}$ and $\F_{\text{D}}$ respectively, are given by
\begin{align}
\lbrace K_{\text{E}} \otimes K_{\text{D}} & : K_{\text{E}} \in \K_{\text{E}}, K_{\text{D}} \in \K_{\text{D}} \rbrace \label{qec:cs:k}\\
& \cup \lbrace F_{\text{E}} \otimes F_{\text{D}}: F_{\text{E}} \in \F_{\text{E}}, F_{\text{D}} \in \F_{\text{D}}\rbrace. \label{qec:cs:f}
\end{align}
We emphasize that if the decoding and encoding quantum error codes are fixed then one can use Theorem 1 of \cite{Nest} to transform the stabilizers (\ref{qec:cs:k}) and (\ref{qec:cs:f}) to a graph state. Furthermore the decoding and re-encoding operation is done virtually, as the final resource state performs both tasks at the same time. This ensures a built-in error correction. Furthermore the delivered quantum information is never unprotected. \newline
An example of a graph state for a code switcher is shown in Fig. \ref{fig:codeswitchers}.

\begin{figure}[h!]
\scalebox{0.5}{
\includegraphics{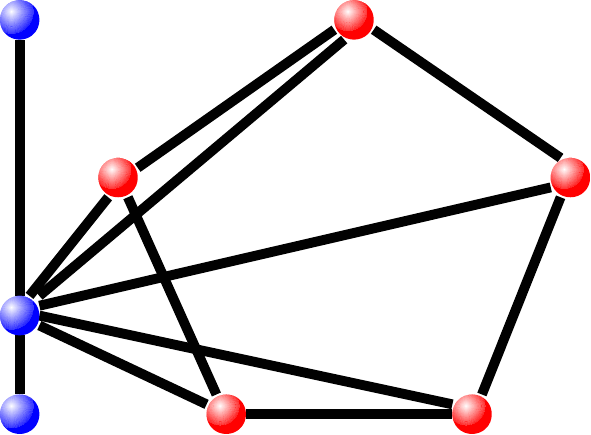}
}
\caption[h!]{Graph state representation (up to local Clifford operations) for a code switcher between the three qubit phase-flip code and the five qubit cluster-ring code. Colors corresponds to input and output qubits respectively.}\label{fig:codeswitchers}
\end{figure}

Last we consider resource states for error syndrome readout. The stabilizers of a resource state for error syndrome readout of a particular quantum error correction is again given (\ref{qec:cs:f}) with the restriction $\F_{\text{D}} = \F_{\text{E}}$ and $\K_{\text{D}} = \K_{\text{E}}$. An example of a local Clifford equivalent graph state for a syndrome readout is depicted in Fig. \ref{fig:synd:bitflip}.

\begin{figure}[h!]
\scalebox{0.7}{
\includegraphics{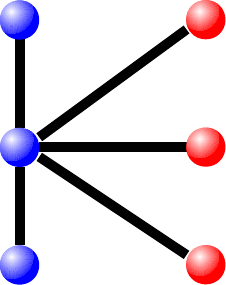}
}
\caption[h!]{LC equivalent graph state for syndrome readout for the 3 qubit bit-flip code. The red vertices correspond to the input qubits whereas the blue vertices to the output qubits.}\label{fig:synd:bitflip}
\end{figure}

\subsection{Entanglement purification at a logical level}
In this last section we consider entanglement purification at a logical level \cite{Zhou}. This is an interesting quantum task, as it combines the benefits of entanglement purification and quantum error correction codes, with potential applications in long-distance quantum communication.

The construction of the resource state which implements such a complex quantum task becomes simple in the proposed framework, as we derived all necessary results in the previous sections. We construct the stabilizers of the resource state for $n$ rounds of entanglement purification at a logical level as follows: \newline
First we virtually decode one part of the logical Bell-pair to be purified in a measurement-based way via the decoding resource state $\ket{\psi_{\mathrm{D}}}$. This yields the sets of auxiliary operators $\K_{\text{D}}$ and $\F_{\text{D}}$ associated with $\ket{\psi_{\mathrm{D}}}$. \newline
The decoding is followed by $n$ virtual rounds of the DEJMPS protocol. In order to perform $n$ rounds of the DEJMPS protocol we apply the recurrence relations  (\ref{dejmps:krec1})--(\ref{dejmps:frec}) $n$ times with  initial auxiliary operator sets $\K_{\text{D}}$ and $\F_{\text{D}}$. This leads to the sets $\K_{n}$ and $\F_{n}$ of auxiliary operators. \newline
Finally, we connect the obtained resource state (decoding followed by  entanglement purification) with sets $\K_{n}$ and $\F_{n}$ to an encoding resource state with auxiliary operators $\K_{\text{E}}$ and $\F_{\text{E}}$ via a Bell-measurement. Hence, according to Theorem \ref{thm.coupling}, the resulting resource state for $n$ rounds of entanglement purification at a logical level is stabilized by the family of operators
\begin{align}
\lbrace K_{\text{E}} \otimes K_{n} &: K_{\mathrm{E}} \in \K_{\text{E}}, K_{n} \in \K_{n} \rbrace \notag \\
& \cup \lbrace F_{\text{E}} \otimes F_{n} : F_{\mathrm{E}} \in \F_{\text{E}}, F_{n} \in \F_{n} \rbrace.
\end{align}
The construction scheme is depicted in Fig. \ref{fig.epplogical}.
\begin{figure}[h!]
\scalebox{0.65}{
\includegraphics{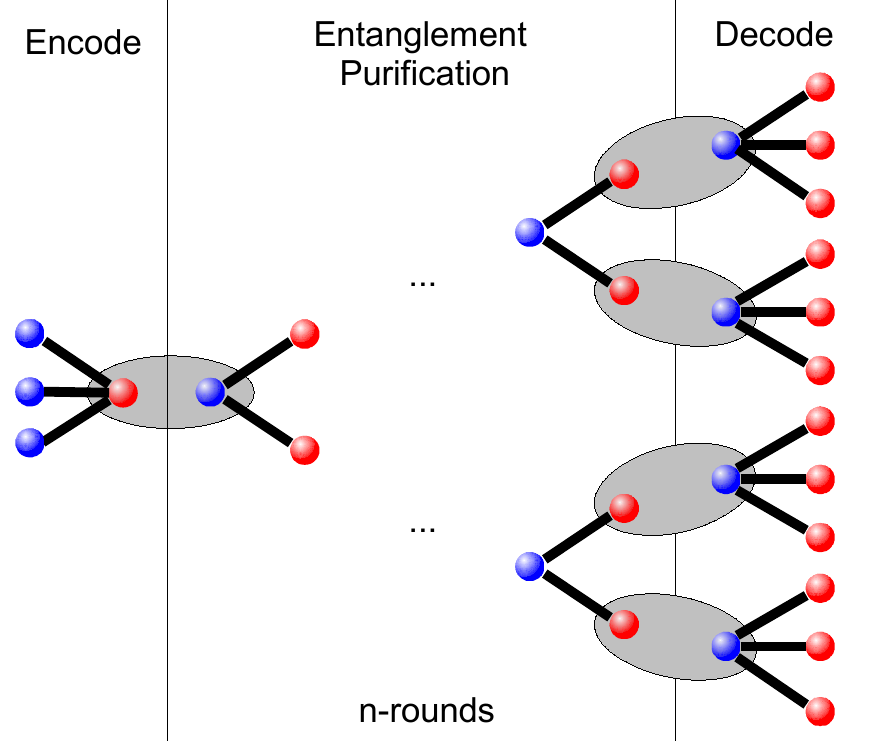}
}
\caption[h!]{Construction of a resource state necessary for $n-$rounds of entanglement purification at a logical level.}\label{fig.epplogical}
\end{figure}
We emphasize that decoding, entanglement purification, encoding and error correction is performed virtually only, as the final resource state performs all three tasks at once. If one needs to construct the resource state for a specific quantum error correction code explicitly, then inserting the appropriate values for $\K_{\text{D}}$, $\F_{\text{D}}$, $\K_{\text{E}}$ and $\F_{\text{E}}$ leads to explicit stabilizers of the resource state. The resulting stabilizers can be transformed via Theorem 1 of \cite{Nest} to a local Clifford equivalent graph state. An example is depicted in Fig. \ref{fig:concretelogicalepp} for the 3 qubit bit-flip code.

\begin{figure}[h!]
\scalebox{1.1}{
\includegraphics{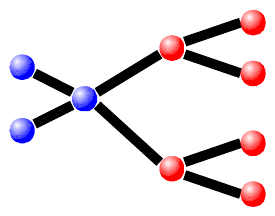}
}
\caption[h!]{The figure shows a LC equivalent graph state for entanglement purification at a logical level, where the three qubit bit-flip code and one round of purification using the DEJMPS protocol are considered. Input qubits are red and output qubits blue.}\label{fig:concretelogicalepp}
\end{figure}

Another interesting approach is to encode a physical qubit into a decoherence free subspace and perform entanglement purification for encoded Bell-pairs, which found promising application in the quantum repeater scheme \cite{ZwergerDFS}. Thereby we define the logical zero and logical one as $\ket{0_L} = \ket{01}$ and $\ket{1_L} = \ket{10}$ respectively. This encoding offers a protection against global dephasing, as phases picked up by $\ket{0}$ and $\ket{1}$ vanish if the physical Bell-pairs are encoded. The resource state necessary to perform one and two rounds of entanglement purification for logical Bell-pairs encoded in a decoherence free subspace are depicted in Fig. \ref{fig:concretelogicalepp:dfs}.

\begin{figure}[h!]
\scalebox{1}{
\includegraphics{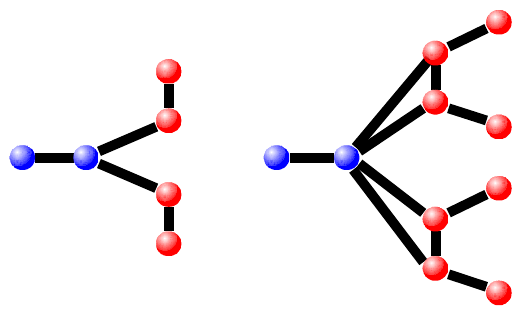}
}
\caption[h!]{The figure shows LC equivalent graph states for entanglement purification at a logical level for one (left) and two (right) purification rounds using a decoherence free subspace encoding \cite{ZwergerDFS}, i.e. the encoding $\ket{0_L} = \ket{01}$ and $\ket{1_L} = \ket{10}$. Input qubits are red and output qubits are blue.}\label{fig:concretelogicalepp:dfs}
\end{figure}

Finally we need to determine the Pauli corrections due to the read-in Bell-measurements. As outlined in Sec. \ref{sec:frm:by} any other outcome than the $\ket{\phip}$ outcome is first propagated through the decoding resource state. Then, the error at the output of the decoding resource state is propagated through the entanglement purification resource state leading to a possibly different Pauli error at the output of the entanglement purification. This error is then propagated through the encoding resource state.

\section{Discussion}\label{sec:discussion}
We proposed a framework which allows us to construct resource states for a measurement-based implementation of concatenated tasks in quantum communication and quantum error correction. In particular, we have  found a closed and analytical description of families of resource states based on recurrence relations of stabilizers. This forms a key tool to generate explicit (graph-) descriptions of the task at hand.

The resource states which we obtain are always of minimal size, containing only input and (if appropriate) output qubits. We derived graph state representations for all resource states, thereby identifying a possible way how to generate these resource states efficiently using commuting two-qubit gates. At the same time, this shows how to generate all these states with high fidelity using multiparty entanglement purification protocols \cite{DurEPP}, and to determine their features under noise and decoherence \cite{HeinGraph}.

We have explicitly constructed resource states for multiple rounds of bipartite entanglement purification protocols, states for encoding, decoding and error syndrome readout for quantum error correction using concatenated error correction codes, code switchers, as well as resource states for combined tasks such as entanglement purification of encoded states. The proposed techniques are however not limited to these examples, but are generically applicable to construct resource states that are combinations of elementary quantum tasks.

The explicit knowledge of required resource states is important in measurement-based implementation of quantum communication \cite{ZwergerQC} and fault-tolerant quantum computation \cite{ZwergerQEC}, in particular to design schemes and methods to prepare them efficiently and with high fidelity in specific setups. Given that such a measurement-based implementation offers very high error thresholds, this may soon become of practical relevance for scalable quantum information processing.

\begin{acknowledgments}
This work was supported by the Austrian Science Fund (FWF): P28000-N27 and SFB F40-FoQus F4012-N16.
\end{acknowledgments}





\appendix

\section{Proof of Main Theorems}\label{app:proofs}

\begin{proof}[Proof of Theorem \ref{thm.koperators}]
Recall that $\ket{G^{n+1}_{i}} = \sum_{k_1,\ldots,k_m} \alpha_{i}(k_1,\ldots,k_m) \ket{G^{n}_{k_1},\ldots,G^{n}_{k_m}}$.   We need to show that $K_{n+1} \ket{G^{n+1}_{i}} = \ket{G^{n+1}_{i \oplus 1}}$ where $K_{n+1} = c \bigotimes\limits^m_{r=1} (F_r)^{\delta_r} \bigotimes\limits^m_{s=1} (K_s)^{\gamma_s}$ for some bit vectors $\gamma$ and $\delta$. We observe $\bigotimes\limits^m_{s=1} (K_s)^{\gamma_s} \ket{G^{n}_{k_1},\ldots,G^{n}_{k_m}} = \ket{G^{n}_{k_1 \oplus \gamma_1},\ldots,G^{n}_{k_m \oplus \gamma_m}}$ which implies
\begin{align}
K_{n+1} & \ket{G^{n+1}_{i}} = c \bigotimes\limits^m_{r=1} (F_r)^{\delta_r} \bigotimes\limits^m_{s=1} (K_s)^{\gamma_s} \notag \\
& \sum\limits^1_{k_1,\ldots,k_m = 0} \alpha_{i}(k_1,\ldots,k_m) \ket{G^{n}_{k_1},\ldots,G^{n}_{k_m}} \notag \\
&= c \bigotimes\limits^m_{r=1} (F_r)^{\delta_r} \sum\limits^1_{k_1,\ldots,k_m=0} \alpha_{i}(k_1,\ldots,k_m) \notag \\
& \ket{G^{n}_{k_1 \oplus \gamma_1},\ldots,G^{n}_{k_m \oplus \gamma_m}} \notag \\
&= c \bigotimes\limits^m_{r=1} (F_r)^{\delta_r} \sum\limits^1_{k_1,\ldots,k_m=0} \alpha_{i}(k_1 \oplus \gamma_1,\ldots,k_m \oplus \gamma_m) \notag \\
& \ket{G^{n}_{k_1},\ldots,G^{n}_{k_m}} \notag \\
&= \sum\limits^1_{k_1,\ldots,k_m=0} c (-1)^{\sum_r \delta_r k_r} \alpha_{i}(k_1 \oplus \gamma_1,\ldots,k_m \oplus \gamma_m) \notag \\
& \ket{G^{n}_{k_1},\ldots,G^{n}_{k_m}}
\end{align}
Now we show the existence of the bit vectors $\gamma$ and $\delta$ as claimed in the theorem. First we observe that
\begin{align}
\alpha_{i \oplus 1} & (k_1,\ldots,k_m) = ^{\otimes m}\langle \phip | k^x_1,\ldots,k^x_m, G^{1}_{i \oplus 1} \rangle \notag \\
& = ^{\otimes m}\langle \phip | (\id \otimes K_1) | k^x_1,\ldots,k^x_m, G^{1}_{i} \rangle \notag \\
& = \left(\prod\limits^m_{k=1} (-1)^{i_k j_k}\right)\, ^{\otimes m} \langle \phip | (K_1 \otimes \id) | k^x_1,\ldots,k^x_m, G^{1}_{i} \rangle \notag \\
& = ^{\otimes m}\langle \phip | (K_1^{T} \otimes \id) | k^x_1,\ldots,k^x_m, G^{1}_{i} \rangle \label{supp:eq:kthm:1}
\end{align}
by using that $K_1$ is an element of the Pauli group $P_n$ and $(\sigma_{i,j} \otimes \id) \ket{\phip} = (-1)^{ij}(\id \otimes \sigma_{i,j}) \ket{\phip}$. Recall that $Z \ket{k^x} = \ket{(k \oplus 1)^x}$, $X \ket{k^x} = (-1)^k \ket{k^x}$ and $Y \ket{k^x} = -i ZX \ket{k^x} = (-i)(-1)^k \ket{(k \oplus 1)^x}$. Thus we have $\sigma_{i,j} \ket{k^x} = (-i)^{ij} (-1)^{kj} \ket{(k \oplus i)^x}$. Recall that $K_1 = a \bigotimes^m_{k=1} \sigma_{i_k,j_k}$ where $a \in \lbrace \pm 1, \pm i \rbrace$, $i_k \in \lbrace 0,1 \rbrace$ and $j_k \in \lbrace 0,1 \rbrace$. This implies for (\ref{supp:eq:kthm:1}) by using $Y^T = - Y$ that
\begin{align}
\alpha_{i \oplus 1} & (k_1,\ldots,k_m) =  ^{\otimes m}\langle \phip | K_1^{T} |k^x_1,\ldots,k^x_m \rangle | G^{1}_{i} \rangle \notag \\
& = a \prod\limits^m_{l=1} (-i)^{i_l j_l} (-1)^{i_l k_l} (-1)^{\sum_{r} k_r j_r} \notag \\
& \, (^{\otimes m}\langle \phip | (k_1 \oplus i_1)^x,\ldots,(k_m \oplus i_m)^x, G^{1}_{i} \rangle) \notag \\
&= a \prod\limits^m_{l=1} i^{i_l j_l} (-1)^{\sum_{r} k_r j_r} \alpha_{i}(k_1 \oplus i_1,\ldots,k_m \oplus i_m).
\end{align}
Thus we have $\gamma = (i_1,..,i_m)$, $\delta = (j_1,..,j_m)$ and $c = a \prod\limits^m_{l=1} i^{i_l j_l}$ which completes the proof.
\end{proof}

\begin{proof}[Proof of Theorem \ref{thm.foperators1}]
Recall that $\ket{G^{n+1}_{i}} = \sum_{k_1,\ldots,k_m} \alpha_{i}(k_1,\ldots,k_m) \ket{G^{n}_{k_1},\ldots,G^{n}_{k_m}}$. Thus we need to show that $F_{n+1} \ket{G^{n+1}_{i}} = (-1)^{i} \ket{G^{n+1}_{i}}$ for $F_{n+1} = c \bigotimes\limits^m_{r=1} (F_r)^{\eta_r} \bigotimes\limits^m_{s=1} (K_s)^{\epsilon_s}.$ Appyling $F_{n+1}$ to $\ket{G^{n+1}_{i}}$ gives
\begin{align}
F_{n+1} & \ket{G^{n+1}_{i}} = c \bigotimes\limits^m_{r=1} (F_r)^{\eta_r} \bigotimes\limits^m_{s=1} (K_s)^{\epsilon_s} \notag \\
& \sum^1_{k_1,\ldots,k_m=0} \alpha_{i}(k_1,\ldots,k_m) \ket{G^{n}_{k_1},\ldots,G^{n}_{k_m}} \notag \\
&= \bigotimes\limits^m_{r=1} (F_r)^{\eta_r} \sum^1_{k_1,\ldots,k_m=0} c \alpha_{i}(k_1,\ldots,k_m) \notag \\
& \ket{G^{n}_{k_1 \oplus \epsilon_1},\ldots,G^{n}_{k_m \oplus \epsilon_m}} \notag \\
&= \bigotimes\limits^m_{r=1} (F_r)^{\eta_r} \sum^1_{k_1,\ldots,k_m=0} c \alpha_{i}(k_1 \oplus \epsilon_1,\ldots,k_m \oplus \epsilon_m) \notag \\
& \ket{G^{n}_{k_1},\ldots,G^{n}_{k_m}} \notag \\
&= \sum^1_{k_1,\ldots,k_m=0} c (-1)^{\sum_r \eta_r k_r} \alpha_{i}(k_1 \oplus \epsilon_1,\ldots,k_m \oplus \epsilon_m) \notag \\
& \ket{G^{n}_{k_1},\ldots,G^{n}_{k_m}}
\end{align}
Hence we need to show that there exist bit vectors $\epsilon$ and $\eta$ such that
\begin{align}
c (-1)^{\sum_r \eta_r k_r} \alpha_{i} & (k_1 \oplus \epsilon_1,\ldots,k_m \oplus \epsilon_m) \notag \\
& = (-1)^{i} \alpha_{i}(k_1,\ldots,k_m). \label{supp:eq:fthm:0}
\end{align}
Recall that $\ket{G^1_i} = (-1)^{i} F_1 \ket{G^1_i}$. Thus we obtain in similar fashion to the proof of Theorem \ref{thm.koperators} that
\begin{align}
\alpha_{i} & (k_1,\ldots,k_m) = ^{\otimes m}\langle \phip | k^x_1,\ldots,k^x_m, G^{1}_{i} \rangle \notag \\
& = (-1)^{i} \, ^{\otimes m}\langle \phip | (\id \otimes F_1) | k^x_1,\ldots,k^x_m, G^{1}_{i} \rangle  \notag \\
&= (-1)^i \left(\prod\limits^m_{k=1} (-1)^{i_k j_k}\right) \notag \\
& \, ^{\otimes m}\langle \phip | (F_1 \otimes \id) | k^x_1,\ldots,k^x_m, G^{1}_{i} \rangle) \notag \\
& = (-1)^i (^{\otimes m}\langle \phip | (F_1^{T} \otimes \id) | k^x_1,\ldots,k^x_m, G^{1}_{i} \rangle) \label{supp:eq:fthm:1}
\end{align}
because $F_1 = a \bigotimes^m_{k=1} \sigma_{i_k,j_k}$ where $a \in \lbrace \pm 1, \pm i \rbrace$, $i_k \in \lbrace 0,1 \rbrace$ and $j_k \in \lbrace 0,1 \rbrace$. Hence we obtain as in the proof of Theorem \ref{thm.koperators} for (\ref{supp:eq:fthm:1})
\begin{align}
\alpha_{i}& (k_1,\ldots,k_m) = (-1)^{i} (^{\otimes m}\langle \phip | F_1^{T} | k^x_1,\ldots,k^x_m \rangle |G^{1}_{i} \rangle ) \notag \\
& = a (-1)^i \prod\limits^m_{l=1} i^{i_l j_l} (-1)^{\sum_{r} k_r j_r} \notag \\
& (^{\otimes m}\langle \phip | (k_1 \oplus i_1)^x,\ldots,(k_m \oplus i_m)^x, G^{1}_{i} \rangle) \notag \\
&= a (-1)^i \prod\limits^m_{l=1} i^{i_l j_l} (-1)^{\sum_{r} k_r j_r} \alpha_{i}(k_1 \oplus i_1,\ldots,k_m \oplus i_m),
\end{align}
i.e. $(-1)^{i} \alpha_{i}(k_1,\ldots,k_m) = a \prod\limits^m_{l=1} i^{i_l j_l} (-1)^{\sum_{r} k_r j_r} \alpha_{i}(k_1 \oplus i_1,\ldots,k_m \oplus i_m)$. Thus (\ref{supp:eq:fthm:0}) holds with $\epsilon = (i_1,..,i_m)$, $\eta = (j_1,..,j_m)$ and $c = a \prod\limits^m_{l=1} i^{i_l j_l}$ which completes the proof.
\end{proof}

\section{Linear independence of stabilizers for special cases}\label{app:lin}

In the following we investigate the linear independence of the sets of auxiliary operators $\K_{n+1}$ and $\F_{n+1}$ obtained via Theorem \ref{thm.koperators} and \ref{thm.foperators1} for special assumptions. Those assumptions are met by the DEJMPS protocol \cite{Deutsch} and the quantum error correction codes studied in the main text. But before we start, we formulate the following Lemma.
\begin{lemma}\label{lem:lin_indepedent}
Assume that the family of operators $\lbrace A_i \rbrace^n_{i=1}$ is linear independent and let $m \in \N$. Then we can construct $nm - m + 1$ linear independent operators of the form $A_{i_1} \otimes .. \otimes A_{i_m}$.
\end{lemma}
\begin{proof}
Consider the following operators (where we have omitted the tensor product symbols for ease of notation):
\begin{align}
\begin{matrix}
A_1 & A_1 & .. & A_1 & A_1 \\
A_1 & A_1 & .. & A_1 & A_2 \\
.. & .. & .. & .. & .. \\
A_1 & A_1 & .. & A_1 & A_n \\
A_1 & A_1 & .. & A_2 & A_1 \\
.. & .. & .. & .. & .. \\
A_1 & A_1 & .. & A_n & A_1 \\
.. & .. & .. & .. & .. \\
A_2 & A_1 & .. & A_1 & A_1 \\
.. & .. & .. & .. & .. \\
A_n & A_1 & .. & A_1 & A_1
\end{matrix}
\end{align}
This family of operators is linear independent by construction. In total we have $m(n-1) + 1 = mn - m +1$ operators, which shows the claim.
\end{proof}
Suppose we have a particular resource state $\ket{\psi^{\O}}$ and the corresponding sets of initial auxiliary operators
\begin{align}
\K_1 &= \left\lbrace K = \bigotimes^{m}_{k=1} \sigma_{\gamma^{(i)}_k, \delta^{(i)}_k}: \, K \text{ satisfies (\ref{def:k:aux})}  \right \rbrace \\
\F_1 &= \left\lbrace F = \bigotimes^{m}_{k=1} \sigma_{\epsilon^{(i)}_k, \eta^{(i)}_k}: \, F \text{ satisfies (\ref{def:f:aux})}  \right \rbrace.
\end{align}
The superscripts in $\gamma, \delta, \epsilon$ and $\eta$ denote different vectors whereas subscripts denote positions within those vectors. Now we want to explicitly construct linear independent stabilizers of a concatenated quantum task.
We assume that the families of operators $\K_n$ and $\F_n$ are linear independent. Then, according to the construction of linear independent operators in Lemma \ref{lem:lin_indepedent}, we find for the families of operators (\ref{eq:k:fam}) and (\ref{eq:f:fam}) that
\begin{align}
|\K_{n+1}| &= \sum^{|\K_1|}_{i = 1} (|\supp \gamma^{(i)}| |\K_n| - |\supp \gamma^{(i)}| + 1) \notag \\
& (|\supp \delta^{(i)}| |\F_n| - |\supp \delta^{(i)}| + 1), \label{eq:lin:ind:k} \\
|\F_{n+1}| &= \sum^{|\F_1|}_{j = 1} (|\supp \epsilon^{(j)}| |\K_n| - |\supp \epsilon^{(j)}| + 1) \notag \\
& (|\supp \eta^{(j)}| |\F_n| - |\supp \eta^{(j)}| + 1). \label{eq:lin:ind:f}
\end{align}
The equations (\ref{eq:lin:ind:k}) and (\ref{eq:lin:ind:f}) will enable us to prove the linear independence for specific quantum tasks.
\begin{theo}\label{thm:lin:ind:1}
Suppose $\gamma^{(i)}=(0,..,0)$ for all $i$, $|\delta| = m$ where $\supp \delta^{(i)} = 1$ for all $i$, $\eta^{(i)} = (0,..,0)$ for all $i$ and $|\epsilon| = 1$ where $|\supp \epsilon| = m$. Furthermore let $|\K_n| + |\F_n| = n+1$. Then we have $nm + 1$ linear independent stabilizers via Lemma \ref{lem:lin_indepedent}.
\end{theo}
\begin{proof}
The assumption $\gamma^{(i)}=(0,..,0)$ for all $i$ implies that all operators in $\K_{n+1}$ will contain only operators from $\F_{n}$. From $\supp \delta^{(i)} = 1$ we further observe that all operators in $\K_{n+1}$ will have exactly one $F_{n} \in \F_{n}$ in the tensor product. Because $|\delta| = m$ and $\gamma^{(i)}=(0,..,0)$ for all $i$ we have that $|\K_1| = |\delta| = m$.  \newline
Similarly, the assumption $\eta^{(i)} = (0,..,0)$ for all $i$ implies that all operators in $\F_{n+1}$ will contain only operators from $\K_{n}$ and $|\epsilon| = 1$ where $|\supp \epsilon| = m$ that all operators in $\F_{n+1}$ will be of the form $F_{n+1} = \bigotimes^{m}_{i=1} K_{i}$ where $K_{i} \in \K_{n}$. Because $\eta^{(i)} = (0,..,0)$ for all $i$ and $|\epsilon| = 1$ we have that $|\F_1| = |\epsilon| = 1$. \newline
Thus (\ref{eq:lin:ind:k}) and (\ref{eq:lin:ind:f}) simplify to
\begin{align}
|\K_{n+1}| &= \sum_{i} |\F_n| = m |\F_n| \\
|\F_{n+1}| &= m |\K_n| - m + 1
\end{align}
Hence $|\K_{n+1}| + |\F_{n+1}| = m (|\K_n| + |\F_n|) - m + 1 = mn + 1$. The linear independence of the families $\K_{n+1}$ and $\F_{n+1}$ follows from the construction of Lemma \ref{lem:lin_indepedent} and the linear independence of the families $\K_{n}$ and $\F_{n}$.
\end{proof}
In Theorem \ref{thm:lin:ind:1} one may also exchange the assumptions on $\gamma$ and $\delta$ as well as the assumptions on $\epsilon$ and $\eta$. Theorem \ref{thm:lin:ind:1} especially applies to the bit-flip code and the phase-flip code. \newline
Furthermore we have the following:
\begin{theo}\label{thm:lin:ind:2}
Suppose $|\K_1| = m$, $|\supp \gamma^{(i)}| = |\supp \delta^{(i)} |= 1$ for all $1 \leq i \leq m$, $|\eta| = 1$ and let $|\F_n| = 1$ for all $n$ where $F_{n+1} = K^{\otimes m}_{n}$ and $K_{n} \in \K_{n}$. Then have $nm + 1$ linear independent stabilizers via Lemma \ref{lem:lin_indepedent}.
\end{theo}
\begin{proof}
The assumption $|\supp \gamma^{(i)}| = |\supp \delta^{(i)} |= 1$ for all $i$ implies that all operators in $\K_{n+1}$ will contain exactly one $K_{n}$ and one $F_{n}$, which are not necessarily at the same position within the tensor product. Because $|\K_{n}| + |\F_{n}| = n +1$ and $|\F_{n}| = 1$ for all $n$ we find $|\K_{n}| = n$. \newline
Thus (\ref{eq:lin:ind:k}) simplifies to
\begin{align}
|\K_{n+1}| = \sum_{i} & |\K_n| |\F_n| = |\gamma| |\K_n| |\F_n| = |\K_1| |\K_n|
\end{align}
Hence $|\K_{n+1}| + |\F_{n+1}| = mn +1$. The linear independence follows from the linear independence of $\K_{n}$ and $\F_{n}$ and the construction via Lemma \ref{lem:lin_indepedent}.
\end{proof}
Theorem \ref{thm:lin:ind:2} is applicable to the DEJMPS protocol. Finally we treat the linear independence of stabilizers constructed via Theorem \ref{thm.coupling}, crucial in the construction of resource states for quantum repeaters, code switchers and entanglement purification at a logical level.

\begin{theo}
Suppose we are given two resource states, $\ket{\psi_1} = \ket{+}_{\mathrm{in}} \ket{G_0} + \ket{-}_{\mathrm{in}} \ket{G_1}$ and $\ket{\psi_2} = \ket{+}_{\mathrm{out}} \ket{H_0} + \ket{-}_{\mathrm{out}}\ket{H_1}$ where $\ket{\psi_1}$ has $n$ output qubits and $\ket{\psi_2}$ has $m$ input qubits. Furthermore, assume that we are given the corresponding sets of auxiliary operators $\K_1, \F_1$ and $\K_2, \F_2$ where $|\K_1| + |\F_1| = n + 1$ and $|\K_2| + |\F_2| = m + 1$. If we connect the resource states $\ket{\psi_1}$ and $\ket{\psi_2}$ via Theorem 3 and each family $\K_1, \F_1$ and $\K_2, \F_2$ is linear independent, then we find $n+m$ linear independent stabilizers.
\end{theo}
\begin{proof}
Similar to Lemma \ref{lem:lin_indepedent} we find a family of operators $\lbrace K_1 \otimes K_2 \rbrace_{K_1 \in \K_1, K_2 \in \K_2}$ which is linear independent and contains $|\K_1| + |\K_2| - 1$ elements. This family can be explicitly constructed via
\begin{align}
\begin{matrix}
A_1 & B_1 \\
.. & .. \\
A_1 & B_{|\K_2|} \\
A_2 & B_1 \\
.. & .. \\
A_{|\K_1|} & B_1 \\
\end{matrix}
\end{align}
where $A_i \in \K_1$ and $B_j \in \K_2$. Furthermore we similarly find $|\F_1| + |\F_2| - 1$ linear independent stabilizers via the family $\lbrace F_1 \otimes F_2 \rbrace_{F_1 \in \F_1, F_2 \in \F_2}$. Because the families $\K_1$ and $\F_1$ and $\K_2$ and $\F_2$ are linear independent by assumption we have in total
\begin{align}
|\K_1| + & |\K_2| - 1 + |\F_1| + |\F_2| - 1 \notag \\
& = n + 1 + m +1 - 2 = n + m
\end{align}
linear independent stabilizers which completes the proof.
\end{proof}

\section{Sufficient number of linear independent stabilizer for general concatenations}\label{app:numberofstabilizers}

In Appendix \ref{app:lin} we already showed that the recurrence relations \eqref{eq:k:fam} and \eqref{eq:f:fam} provide a sufficient number of linear independent stabilizers for some special cases of interest. Here we show that this is indeed a general feature of this approach not limited to examples that have a quickly recognisable way to pick a complete set of stabilizers.

For this purpose we split up one concatenation step into smaller parts, namely performing only one projection on $\ket{\phip}$ at a time instead of all at once. A slight modification of the proofs in Appendix \ref{app:proofs} shows that this leads to only the $i$-th block in the operators in $\K_1$ and $\F_1$ being updated with elements of $\K_n$ and $\F_n$ instead of all at once where $i$ is the qubit in the set $\mathrm{out}$ on which part of the projection acts.

If the initial state $\ket{\psi^\O}$ consists of $M$ qubits and $\ket{\psi^{\O'}}$ is given by $N$ qubits, we need $N+M-2$ linear independent stabilizers to describe the state after the first projection is performed. The stabilizers obtained from $\K_1$ and $\F_1$ can be categorized by the Pauli operator that is going to be replaced, namely
\begin{align}
 \begin{matrix}
 m_X \text{ of the type} & . & \otimes & X & \otimes & \dots \\
 m_Z \text{ of the type} & . & \otimes & Z & \otimes & \dots \\
 m_Y \text{ of the type} & . & \otimes & Y & \otimes & \dots \\
 m_I \text{ of the type} & . & \otimes & \id & \otimes & \dots
 \end{matrix}
\end{align}
with $m_X + m_Y + m_Z + m_I = M$. For this analysis it does not matter whether the stabilizers relate to $K$-type or $F$-type auxiliary operators. If the marked qubit is entangled with the rest of the qubits, at least two of $m_X$, $m_Y$, $m_Z$ are non-zero.

According to Theorem \ref{thm.koperators} and \ref{thm.foperators1} the projection replaces the $X$ [$Z$,$Y$] in the $m_X$ [$m_Z$,$m_Y$] $X$ [$Z$,$Y$] type stabilizers by $F_n$ [$K_n$,$iK_nF_n$] operators from the set(s) $\F_n$ [and/or $\K_n$]. For $\id$ type stabilizers the identity is simply replaced by the identity on a larger space and the number of stabilizers $m'_I = m_I$ is unaffected. From Lemma \ref{lem:lin_indepedent} it is clear that if $m_X \neq 0$ we obtain $m'_X = m_X + n_F - 1 $ stabilizers that are independent from one another, where $n_F$ is $|\F_n|$. Similarly if $m_Z \neq 0$ we obtain $m'_Z = m_Z + n_K - 1$ new linear independent stabilizers, with $n_K = |\K_n|$.

For the $Y$ type stabilizers it is less straightforward to see, but a similar technique as in Lemma \ref{lem:lin_indepedent} can be used also in this case. If $m_Y \neq 0$ we get $m'_Y = m_Y + n_F + n_K - 2$ stabilizers that are independent from one another, but these are not linear independent from those obtained by either $X$ or $Z$ type stabilizers. If $m_X \neq 0$ [$m_Z \neq 0$] we reduce that previous number $m'_Y$ by $n_F - 1$ [$n_K - 1$]. From that we see that the new number of independent stabilizers $m'_X + m'_Y + m'_Z + m'_I = M + N - 2$. (Recall that $n_F + n_K = N$.) From this it follows that we always obtain a full set of linear independent stabilizers for one projection step and applying this argument repeatedly means that we always have a sufficient number of linear independent stabilizers to uniquely describe the resulting state.

\section{Five-qubit cluster ring code}\label{app:cluster}

In the following we provide the recurrence relations of the five-qubit cluster ring code. This quantum error correction code is very interesting because of its small size and error detecting properties. The five-qubit cluster ring code is able to correct single qubit phase-flip and bit-flip errors. The resource state for its measurement-based implementation is a graph state, see Fig. \ref{fig:clusterring}, stabilized by the following operators:
\begin{align}
\begin{matrix}
Z & X & Z & I & I & Z \\
Z & Z & X & Z & I & I \\
Z & I & Z & X & Z & I \\
Z & I & I & Z & X & Z \\
Z & Z & I & I & Z & X \\
X & Z & Z & Z & Z & Z
\end{matrix} \label{eq:qec:cluster}
\end{align}

\begin{figure}[h!]
\scalebox{0.5}{
\includegraphics{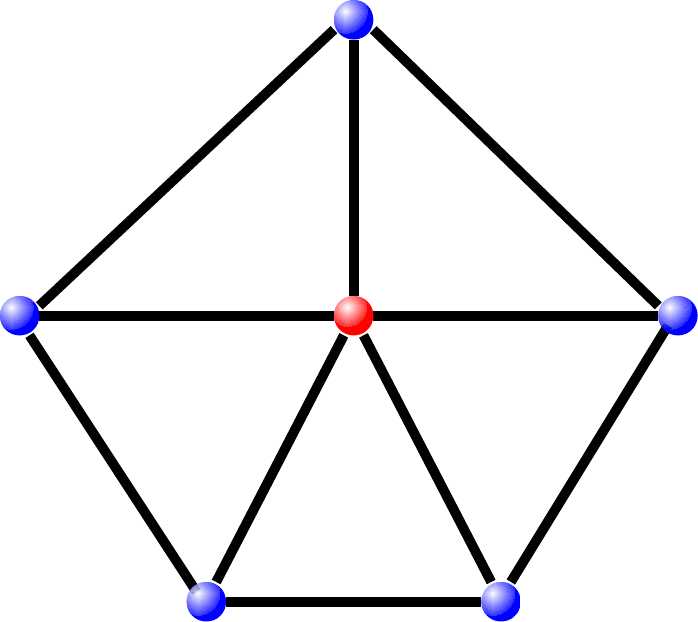}
}
\caption[h!]{Resource state of the five-qubit clusterring code. The red vertex is the input qubit.}\label{fig:clusterring}
\end{figure}

From (\ref{eq:qec:cluster}) we easily find the initial sets of auxiliary operators $\K_1$ and $\F_1$. Hence Theorem \ref{thm.koperators} and \ref{thm.foperators1} implies that
\begin{align}
\K_{n+1} &= \lbrace F \otimes K_{1} \otimes \id \otimes \id \otimes K_{2} : \notag \\
 & K_{1} \in \K_{n}, K_{2} \in \K_{n}, F \in \F_{n} \rbrace \notag \\
 & \cup \lbrace K_{1} \otimes F \otimes K_{2} \otimes \id \otimes \id : \notag \\
 & K_{1} \in \K_{n}, K_{2} \in \K_{n}, F \in \F_{n} \rbrace \notag \\
 & \cup \lbrace \id \otimes K_{1} \otimes F \otimes K_{2} \otimes \id : \notag \\
 & K_{1} \in \K_{n}, K_{2} \in \K_{n}, F \in \F_{n} \rbrace  \notag \\
 & \cup \lbrace \id \otimes \id \otimes K_{1} \otimes F \otimes K_{2} : \notag \\
 & K_{1} \in \K_{n}, K_{2} \in \K_{n}, F \in \F_{n} \rbrace  \notag \\
 & \cup \lbrace K_{1} \otimes \id \otimes \id \otimes K_{2} \otimes F : \notag \\
 & K_{1} \in \K_{n}, K_{2} \in \K_{n}, F \in \F_{n} \rbrace \\
\F_{n+1} &= \lbrace K_{1} \otimes K_{2} \otimes K_{3} \otimes K_{4} \otimes K_{5} : \notag \\
& \forall 1 \leq i \leq 5: K_i \in \K_n \rbrace .
\end{align}

Hence the recurrence relations above can be used to concatenate the five-qubit cluster ring code with any other quantum tasks proposed within this paper.

\section{Graph state representation of the resource state for a generalized Shor code}\label{app:genshor}

Recall that the stabilizers of a $[m_1,m_2]-$generalized Shor code are given by
\begin{align}
\lbrace Z & \otimes (X^{\otimes {m_2}})^{(j)}: 1 \leq j \leq m_1 \rbrace \label{qec:genshor:stab:1:supp}\\
& \cup \left\lbrace X \otimes \bigotimes\limits^{m_1 - 1}_{k = 0} Z^{(i_k + {m_2} k)}: 1 \leq i_k \leq m_2 \right\rbrace. \label{qec:genshor:stab:2:supp}
\end{align}
Now we transform the stabilizers (\ref{qec:genshor:stab:1:supp}) and (\ref{qec:genshor:stab:2:supp}) to graph state stabilizers by applying local unitaries. First we observe that not all obtained stabilizers are linear independent. Hence we need to find a subset of stabilizers which is linear independent. For stabilizers of type (\ref{qec:genshor:stab:1:supp}) we have
\[
\begin{array}{c|ccc|ccc|ccc}
Z & X & ... & X & I & ... & ... & ... & ... & I \\
Z & I & ... & I & X & ... & X & I & ... & I \\
... & ... & ... & ... & ... & ... & ... & ... & ... & ... \\
Z & I & ... & ... & ... & ... & I & X & ... & X \\
\end{array}
\]
where each block of $X$ operators is of length $m_2$. This yields $m_1$ linear independent stabilizers of type (\ref{qec:genshor:stab:1:supp}). To find linear independent stabilizers of type (\ref{qec:genshor:stab:2:supp}) we use Lemma \ref{lem:lin_indepedent}. According to Lemma \ref{lem:lin_indepedent} we find that the stabilizers
\setcounter{MaxMatrixCols}{20}
\[
\begin{array}{c|ccccc|ccccc|c|ccccc}
X & Z & I & ... & ... & I & Z & I & ... &... & I & ... & Z & I & ... & ... & I \\
... & ... & ...& ... & ... & ... & ... & ... & ... & ... & ... & ... & ... & ... & ... & ... & ... \\
X & Z & I & ... & ... & I & Z & I & ... & ... & I & ... & I & ... & ...& I & Z \\
\hline X & Z & I & ... & ... & I & I & Z & I & ... & I & ... & Z & I & ... & ... & I \\
... & ... & ... & ... & ... & ... & ... & ... & ... & ... & ... & ... & ... & ... & ... & ... & ... \\
X & Z & I & ... & ... & I & I & ...  & ... & I & Z & ... & Z & I & ... & ... & I \\
\hline X & I & Z & I & ... & I & Z & I & ... & ... & I & ... & Z & I & ... & ... & I \\
... & ... & ... & ..& ... & ... & ... & ... & ... & ... & ... & ... & ... & ... & ... & ... & ... \\
X & I & ... & ... & I & Z & Z & I & ... & ... & I & ... & Z & I & ... & ... & I \\
\end{array}
\]
which are of type (\ref{qec:genshor:stab:2:supp}) and linear independent. Hence we observe that we have $1 + m_1(m_2 - 1) + m_1 = 1 + m_1 m_2$ linear independent stabilizers as required to completely describe the resulting graph state. To summarize, the stabilizers of the resource state are given by
\[
\begin{array}{c|ccccc|ccccc|c|ccccc}
Z & X & ... & ... & ... & X & I & ...& ... & ... & ... & ... &... & ... & ... & ... & I \\
Z & I & ... & ... & ... & I & X & ... & ... & ... & X & ... & I & ... & ... & ... & I \\
... & ... & ... & ... & ... & ... & ... & ... & ... & ... & ... & ... & ... &... & ... & ... & ... \\
Z & I & ... & ... &... & ... & ... & ... & ... &... & ... & ... & X & ... & ... & ... & X \\

\hline X & Z & I & ... & ... & I & Z & I & ... &... & I & ... & Z & I & ... & ... & I \\
\hline ... & ... & ...& ... & ... & ... & ... & ... & ... & ... & ... & ... & ... & ... & ... & ... & ... \\
X & Z & I & ... & ... & I & Z & I & ... & ... & I & ... & I & ... & ...& I & Z \\
\hline X & Z & I & ... & ... & I & I & Z & I & ... & I & ... & Z & I & ... & ... & I \\
... & ... & ... & ... & ... & ... & ... & ... & ... & ... & ... & ... & ... & ... & ... & ... & ... \\
X & Z & I & ... & ... & I & I & ...  & ... & I & Z & ... & Z & I & ... & ... & I \\
\hline X & I & Z & I & ... & I & Z & I & ... & ... & I & ... & Z & I & ... & ... & I \\
... & ... & ... & ..& ... & ... & ... & ... & ... & ... & ... & ... & ... & ... & ... & ... & ... \\
X & I & ... & ... & I & Z & Z & I & ... & ... & I & ... & Z & I & ... & ... & I \\
\end{array}
\]

For convenience we define $X' = X \otimes \bigotimes\limits^{m_1 - 1}_{k = 0} Z^{(i_k + m_2 k)}$ where all $i_k$ are set to $1$ (observe that this is the first stabilizer having an $X$ on the input qubit in the previous expression). In order to obtain a graph state we apply a Hadamard gate to the output qubits $i_k + m_2 k$ where each $i_k = 2,...,m_2$ for all $k$. This yields the stabilizers
\[
\begin{array}{c|ccccc|ccccc|c|ccccc}
Z & X & Z & ... & ... & Z & I & ... & ...& ... & ... & ... &... & ... & ... & ... & I \\
Z & I & ... & ... & ... & I & X & Z & ... & ... & Z & ... & I & ... & ... & ...& I \\
... & ... & ... & ... & ... & ... & ... & ... & ... & ... & ... & ... &... & ... & ... & ... & ... \\
Z & I & ... & ... &... & ... & ... & ... &... & ... & ... & ... & X & Z & ... & ... & Z \\

\hline X & Z & I & ... & ... & I & Z & I & ... & ... & I & ... & Z & I & ... & ... & I \\
\hline ... & ... & ... & ... & ... & ... & ... & ... & ... & ... & ... & ... & ... & ... & ... & ... & ... \\
X & Z & I & ... & ... & I & Z & I & ... & ... & I & ... & I & ... & ... & I & X \\
\hline X & Z & I & ... & ... & I & I & X & I &... & I & ... & Z & I & ... & ... & I \\
... & ... & ... & ... & ... & ... & ... & ... & ... & ... & ... & ... & ... & ... & ... & ... & ... \\
X & Z & I & ... & ... & I & I & ... & ... & I & X & ... & Z & I & ... & ... & I \\
\hline X & I & X & I & ... & I & Z & I & ... & ... & I & ... & Z & I & ... & ... & I \\
... & ... & ... & ... & ... & ... & ... & ... & ... & ... & ... & ... & ... & ... & ... & ... & ... \\
X & I & ... & ... & I & X & Z & I & ... & ... & I & ... & Z & I & ... & ... & I \\
\end{array}
\]
Next we multiply all $F$-type stabilizers (except $X'$), i.e. the lower three blocks in the previous expression, with $X'$ and replace them. Hence we end up with the stabilizers
\[
\begin{array}{c|ccccc|ccccc|c|ccccc}
Z & X & Z & ... & ... & Z & I & ... & ...& ... & ... & ... &... & ... & ... & ... & I \\
Z & I & ... & ... & ... & I & X & Z & ... & ... & Z & ... & I & ... & ... & ... & I \\
... & ... & ... & ... & ... & ... & ... & ... & ... & ... & ... & ... &... & ... & ... & ... & ... \\
Z & I & ... & ... & ...& ... & ... & ... &... & ... & ... & ... & X & Z & ... & ... & Z \\
\hline X & Z & I & ... & ... & I & Z & I & ... & ... & I & ... & Z & I & ... & ... & I \\
\hline ... & ... & ... & ... & ... & ... & ... & ... & ... & ... & ... & ... & ... & ... & ... & ... & ... \\
I & I & ... & ... & ... & I & I & ... & ... & ... & I & ... & Z & I & ... & I & X \\
\hline I & I & ... & ... & ... & I & Z & X & I & ... & I & ... & I & ... & ... & ... & I \\
... & ... & ... & ... & ... & ... & ... & ... & ...& ... & ... & ... & ... & ... & ... & ... & ... \\
I & I & ... & ... & ... & I & Z & I & ... & I & X & ... & I & ... & ... & ... & I \\
\hline I & Z & X & I & ... & I & I & ... & ... & ... & I & ... & I & ... & ... & ... & I \\
... & ... & ... & ... & ... & ... & ... & ... & ... & ... & ... & ... & ... & ... & ... & ... & ... \\
I & Z & I & ... & I & X & I & ... & ... & ... & I & ... & I & ... & ... & ... & I \\
\end{array}
\]
We immediately observe that specific output particles (first column $X$ stabilizer) have $m_1$ edges. Furthermore, each of those output qubits has $m_2-1$ edges to the remaining  output qubits. The final graph state is shown in Fig. \ref{fig:genshor:supp}.
\begin{figure}[h!]
\scalebox{0.8}{
\includegraphics{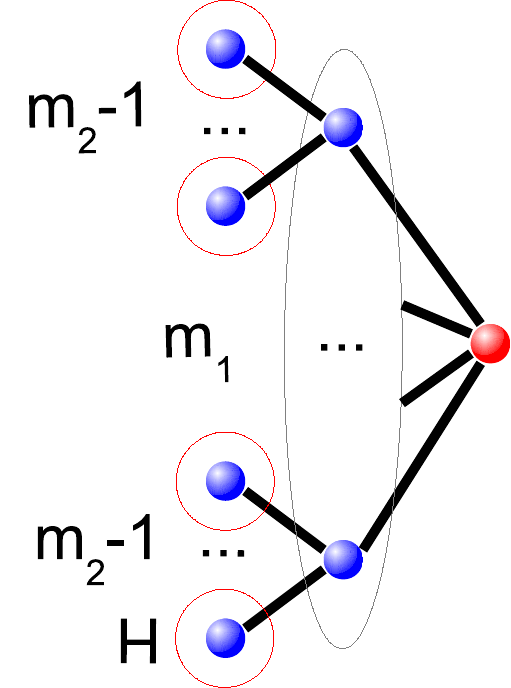}
}
\caption[h!]{Resource state of a $[m_1,m_2]-$generalized Shor code. The red circles indicate that a Hadamard rotation needs to be applied.}\label{fig:genshor:supp}
\end{figure}

\section{Concatenated Shor code}\label{app:concshor}

We also investigated the concatenation of a variant of the Shor code. The stabilizers can be obtained by repeatedly applying the rules of the resource state of the phase-flip code as its function is to switch bases and implement a bit-flip code. The resource state is obtained from 3 concatenations of that code, which is equivalent to concatenating the Shor code once. Its local Clifford equivalent graph state is depicted in Fig. \ref{fig:concatenated_shor}.

\begin{figure}[h!]
 \includegraphics[width=0.5\columnwidth]{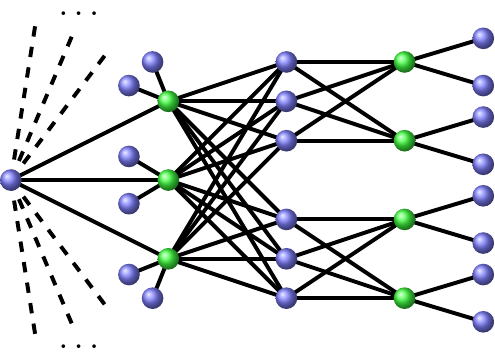}
 \caption{\label{fig:concatenated_shor} The resource state for the concatenated Shor code is local Clifford equivalent to a two-colorable graph state. The qubit on the very left is the input qubit. The dashed
 lines connect to other copies of the center block with 27 qubits.}
\end{figure}

\newpage
\bibliography{refs}

\end{document}